\documentclass[pra,twocolumn,a4paper,nofootinbib,showpacs]{revtex4-1}

\usepackage{hyperref}

\newcommand{\papertitle}{Non-negative subtheories and quasiprobability 
representations of qubits}

\hypersetup{
 pdfauthor={Joel J. Wallman},
	pdftitle={\papertitle}
}
\usepackage{graphicx}
\usepackage{bm}
\usepackage{amsmath}
\usepackage{amssymb}
\usepackage{amsthm}
\usepackage{color}
\usepackage{ulem}

\theoremstyle{plain}
\newtheorem{thm}{Theorem}[section]
\newtheorem{lem}[thm]{Lemma}

\theoremstyle{definition}

\theoremstyle{remark}

\definecolor{nblue}{rgb}{0.2,0.2,0.7}
\definecolor{ngreen}{rgb}{0.1,0.5,0.1}
\definecolor{nred}{rgb}{0.8,0.2,0.2}
\definecolor{nblack}{rgb}{0,0,0}

\newcommand{\unit}{\mathbf{1}}
\newcommand{\supp}[1]{\mathcal{S}\left(#1\right)}
\newcommand{\supps}[1]{\mathcal{S}\left[#1\right]}
\newcommand{\mc}[1]{\mathcal{#1}}

\newcommand{\mbb}[1]{\mathbb{#1}}

\newcommand{\ket}[1]{|#1\rangle}

\newcommand{\tr}[1]{\text{Tr}\left(#1\right)}
\newcommand{\trs}[1]{\text{Tr}\left[#1\right]}

\begin{document}

\title{\papertitle}

\author{Joel J. \surname{Wallman}}
\affiliation{Centre for Engineered Quantum Systems, School of Physics, The 
University of Sydney, Sydney, NSW 2006, Australia}
\author{Stephen D. \surname{Bartlett}}
\affiliation{Centre for Engineered Quantum Systems, School of Physics, The 
University of Sydney, Sydney, NSW 2006, Australia}

\date{\today}

\begin{abstract}
Negativity in a quasiprobability representation is typically interpreted as 
an indication of nonclassical behavior. However, this does not preclude 
states that are non-negative from exhibiting phenomena typically associated 
with quantum mechanics---the single qubit stabilizer states have non-negative 
Wigner functions and yet play a fundamental role in many quantum information 
tasks. We seek to determine what other sets of quantum states and 
measurements for a qubit can be non-negative in a quasiprobability 
representation, and to identify nontrivial unitary groups that permute the 
states in such a set. These sets of states and measurements are analogous to 
the single qubit stabilizer states. We show that no quasiprobability 
representation of a qubit can be non-negative for more than four bases and 
that the non-negative bases in any quasiprobability representation must satisfy 
certain symmetry constraints. We provide an exhaustive list of 
the sets of single qubit bases that are non-negative in some 
quasiprobability representation and are also permuted by a nontrivial unitary group. 
This list includes two families of three bases that both include the 
single qubit stabilizer states as a special case and a family of 
four bases whose symmetry group is the Pauli group. For higher 
dimensions, we prove that there can be no more than $2^{d^2}$ states in 
non-negative bases of a $d$-dimensional Hilbert space in any quasiprobability 
representation. Furthermore, these bases must satisfy certain symmetry 
constraints, corresponding to requiring the bases to be sufficiently 
complementary to each other.
\end{abstract}

\pacs{03.65.Ta, 03.67.-a, 03.65.Sq}

\maketitle

\section{Introduction}

As an alternative to the standard formulation of quantum theory in terms of 
vectors in a Hilbert space, it is possible to express quantum states, 
transformations, and measurements as functions on some state space. The most 
common of such representations is the Wigner function~\cite{Wigner1971}, 
which represents the quantum state of a particle as a distribution over the 
classical phase space of the particle. However, this function cannot be 
interpreted as a probability distribution as it takes on negative values. 
Such descriptions are referred to as \textit{quasiprobability representations}~\cite{Ferrie2011}. 

The occurrence of negative probabilities in the description of a quantum 
state or measurement is often thought of as an indication of 
``quantum-ness''~\cite{Kenfack2004}. Conversely, any state or measurement
that can be described by non-negative (true) probabilities is sometimes 
said to be ``classical''. For the Wigner function of particle mechanics,
Hudson's theorem shows that a pure quantum state has a non-negative 
Wigner function if and only if it is a 
Gaussian state~\cite{Hudson1974,Soto1983}. The discrete Wigner function 
extends these results to finite-dimensional quantum systems. 
For finite odd dimensions, the only pure states with non-negative 
Wigner functions are stabilizer states~\cite{Gross2006,Gross2007}.
Negative distributions, on the other hand, are nonclassical in the 
sense that they are contextual~\cite{Spekkens2008} and can serve as a
resource for quantum computation~\cite{Veitch2012}.

However, preparing a system in a quantum state described by a non-negative 
Wigner function is neither a necessary or sufficient condition to say that 
this system is classical. It is not necessary, because the Wigner function 
is only one possible quasiprobability representation of quantum theory. One 
can construct a quasiprobability representation in which any individual state 
or measurement has a non-negative distribution. It is also not sufficient, as 
it is possible to make all quantum states have non-negative distributions, but 
this forces the conditional probabilities of some measurements to take on 
negative values~\cite{Spekkens2008,Ferrie2008}. To say that a system has a 
classical description, we require the set of preparations, transformations, 
and measurements we are considering to all have nonnegative distributions in 
a quasiprobability representation. Consequently, we aim to determine what 
\textit{subtheories} of quantum mechanics (i.e., theories constructed from 
closed subsets of quantum states, transformations and measurements) are 
non-negative in some quasiprobability representation, or, alternatively, can 
be described ``classically''. 

In this paper, we consider a particular class of subtheories containing 
orthonormal bases both as a set of preparations and as a projective 
measurement. For a fixed quasiprobability representation, if the 
distributions corresponding to the states in an orthonormal basis and the 
projective measurement in the basis are all non-negative, we refer to the 
basis as \textit{non-negative}. 

The question we are then considering is what sets of bases can be 
simultaneously non-negative in some quasiprobability representation. 
Additionally, we would like to determine what groups of unitary 
transformations permute a set of non-negative bases, as it then seems likely 
that circuits composed of such gates can be efficiently simulated. While our 
results are primarily applicable to single systems, we believe that these 
results may be important in the pursuit of additional classically-simulatable 
subtheories of multiple systems. At present, we know of very few nontrivial 
sets of quantum gates that can be efficiently simulated on a classical 
computer, such as the Clifford group and matchgate circuits. Finding other 
subsets of gates that are also efficiently simulatable may provide key 
insights into the nature of quantum computation and the origin of any 
advantage of quantum computation over classical computation.

While our results are derived in the context of single systems, it is 
important to note that several interesting questions remain on the 
``quantum'' nature of single systems, as demonstrated by a range of phenomena 
such as the proof of contextuality for a single qubit~\cite{Spekkens2005}, 
models of quantum computation that use only one ``clean'' qubit~\cite{Shor2008}
and recent results on the distillability of magic states~\cite{Veitch2012}. 
We believe that our results show some promise in answering such questions.

This paper is structured as follows. We begin by introducing quasiprobability 
representations of quantum mechanics in Sec.~\ref{sec:quasi} and prove some 
elementary properties that any quasiprobability representation of quantum 
mechanics must satisfy. We obtain an upper bound on the number of 
non-negative bases of a qubit for any quasiprobability representation and 
relations that any non-negative bases of a qubit must satisfy in 
Sec.~\ref{sec:qubit_bound}. We consider examples of quasiprobability 
representations that are non-negative for subtheories of a qubit with 
nontrivial transformations in Sec.~\ref{sec:examples}. In 
Sec.~\ref{sec:higher_dimensions}, we generalize the theorems in 
Sec.~\ref{sec:qubit_bound} to $d$ dimensions and obtain an upper bound on 
the number of states that are elements of non-negative bases for any 
finite $d$. We conclude with a discussion on the implications for 
quantum computation in Sec.~\ref{negativity:conclusion}.

\section{Ontological models and quasiprobability representations}
\label{sec:quasi}

In this section we motivate the study of quasiprobability representations, 
and non-negative subtheories within them, from a foundational perspective 
based on ontological models. We present our formalism for quasiprobability 
representations, and prove some elementary properties that any such 
representation of quantum mechanics must satisfy.

One approach to explaining the predictions of quantum mechanics is to 
formulate an ontological model (i.e., a hidden variable model) that 
reproduces some or all of the measurement statistics of quantum theory. Such 
a model is defined over some ontic state space $\Lambda$. Preparations of 
quantum states correspond to probability measures $\mu$ over $\Lambda$. 
Measurements correspond to sets of conditional measurement probabilities $\xi(
k|\lambda)$ of observing an outcome $k$ given that a system is in the ontic 
state $\lambda\in\Lambda$, respectively. A desirable feature of such a model 
is a revised notion of noncontextuality~\footnote{While there are ontological 
models that do not satisfy the revised assumption of noncontextuality, e.g., 
the models in~\cite{bb_model,Montina2006}, we do not consider them here.}, so 
that each preparation as a density operator corresponds to a single 
probability distribution and each measurement effect corresponds to a unique 
conditional probability~\cite{Spekkens2005}.

While no noncontextual ontological model can reproduce all of quantum 
mechanics, it may be possible to define a noncontextual ontological model for 
a subtheory. Such a subtheory may still capture some of the essential 
phenomena of quantum mechanics. For example, the set of states, 
transformations, and measurements with Gaussian Wigner distributions in 
particle mechanics can be described by a noncontextual ontological model~\cite{Bartlett2012}. 
For finite odd dimensions, the single qudit subtheory 
consisting of stabilizer states and measurements has an ontological model, 
namely the discrete Wigner function~\cite{Gibbons2004} restricted to this 
set~\cite{Gross2006,Gross2007}. Spekkens' toy theory~\cite{Spekkens2007} also 
provides an ontological model that is in many ways analogous to the 
stabilizer subtheory of a single qubit. Despite being ``classical'', these 
models allow a variety of information processing tasks typically associated 
with quantum mechanics, such as quantum teleportation and dense coding; 
see~\cite{Bennett1993,Bennett1992,Spekkens2007,Weedbrook2012,Bartlett2012}.

An alternate but related approach to explaining quantum mechanical 
predictions is to use a quasiprobability representation. As with an 
ontological model, a quasiprobability representation is defined over a state 
space $\Lambda$ (often, but not necessarily, a classical phase space) which 
can be interpreted as an ontic state space, and it associates preparations 
and measurements with distributions and conditional distributions, 
respectively, over $\Lambda$. A quasiprobability representation is a faithful 
representation of the density operators and measurement effects of quantum 
theory---that is, the map from operators on Hilbert space to distributions on 
$\Lambda$ is linear and injective. Such representations can maintain 
noncontextuality but nevertheless reproduce the quantum predictions because, 
unlike ontological models, the distributions corresponding to such 
preparations and measurements are allowed to take on negative values (thus 
the term ``quasiprobability''). That is, nonnegativity of distributions on 
ontic states is the classical assumption dropped by quasiprobability 
representations in order to reproduce the predictions of quantum theory.

A quasiprobability representation of quantum mechanics cannot be non-negative 
for all preparations and measurements; this fact is equivalent to the fact 
that an ontological model of quantum mechanics cannot be 
noncontextual~\cite{Spekkens2008}. However, a quasiprobability representation 
can be non-negative for a subtheory of quantum mechanics; specifically, the 
preparations, transformations and measurements within the subtheory can all 
possess non-negative probability distributions. In such cases, the quasiprobability 
representation provides a noncontextual ontological model for this subtheory. 
That is, the existence of a noncontextual ontological model for a subtheory 
and a quasiprobability representation that is non-negative for a subtheory 
are equivalent notions~\cite{Spekkens2008}. Throughout this paper, we will 
restrict our language (for the most part) to that of quasiprobability 
representations and the possible existence of non-negative subtheories within 
them, although the reader should keep in mind that results for non-negative 
subtheories in quasiprobability representations apply identically to a 
perspective of noncontextual ontological models for such subtheories.

An example of a quasiprobability representation frequently used in quantum 
optics is the Wigner function. The subtheory of quantum mechanics consisting 
of Gaussian states, transformations and measurements, is completely described 
by a noncontextual ontological model for which the probability distributions 
and conditional measurement probabilities are all 
non-negative~\cite{Bartlett2012}. This subtheory is embedded in a 
quasiprobability representation, describing all possible states, 
transformations, and measurements, but using negative probabilities
for non-Gaussian ones.

A quasiprobability representation of quantum mechanics over a space $\Lambda$ 
is defined by two sets of Hermitian operators, $\{F(\lambda)\}$ and $\{G(
\lambda)\}$, acting on a $d$-dimensional Hilbert space 
$\mc{H}_d$~\cite{Ferrie2009,Spekkens2008}. The sets $\{F(\lambda)\}$ 
and $\{G(\lambda)\}$ are dual frames for the space of operators acting 
on $\mc{H}_d$~\cite{Ferrie2009}. The quasiprobability distribution 
associated with a quantum state $\rho$ is
\begin{align}\label{eq:quasi_F}
\mu_{\rho}(\lambda) &= \trs{\rho F(\lambda)}\in\mathbb{R}	\,.
\end{align}
The \textit{support} of a state $\rho$ is the set
\begin{align}
\supp{\rho} &= \{\lambda\in\Lambda:\mu_{\rho}(\lambda) \neq 0\}	\,. \label{
eq:suppP}
\end{align}
A point $\lambda$ is \textit{compatible} with a quantum state $\rho$ if 
$\lambda\in\supp{\rho}$ and \textit{incompatible} with $\rho$ otherwise. 

For measurements, the conditional quasiprobability of an effect $E$ (i.e., an 
element of a POVM) occurring if the system is in the state $\lambda$ is given 
by an indicator function,
\begin{align}\label{eq:quasi_G}
\xi_{E}(\lambda) &= \trs{E G(\lambda)}\in\mathbb{R}	\,.
\end{align}
As a system is always in some state, we require $\mu$ to be normalized, i.e., 
\begin{align}\label{eq:ontic_normalized_state}
\int_{\Lambda}d\lambda\, \mu_{\rho}(\lambda) = 1
\end{align}
for all states $\rho$. Similarly, as some outcome of a measurement always 
occurs, we require
\begin{align}\label{eq:completeness}
\sum_j \xi_{E_j}(\lambda)=1
\end{align}
for all $\lambda\in\Lambda$ and all POVMs $\{E_j\}$. In order to reproduce 
the Born rule, $\mu$ and $\xi$ must satisfy 
\begin{align}\label{eq:ontic_probability}
\tr{\rho E} = \int_{\Lambda}d\lambda\, \mu_{\rho}(\lambda)\xi_{E}(\lambda)	\,.
\end{align}
for all states $\rho$ and all POVMs $\{E_j\}$.

While negative values of $\mu$ and $\xi$ are allowed, there may be states 
$\rho$ or effects $E$ such that $\mu_{\rho}(\lambda)\geq 0$ or 
$\xi_{E}(\lambda)\in [0,1]$ for all $\lambda\in\Lambda$, respectively. 
Such states and effects are referred to as \textit{non-negative}.

In quantum mechanics, there is a one-to-one correspondence between 
orthonormal bases of $\mc{H}_d$ and projective measurements, as any 
orthonormal basis $\{\rho(j):j\in\mbb{Z}_d\}$ corresponds to a set of 
preparations \textit{and} to the effects for a projective measurement. 
Motivated by this correspondence, we define a \textit{non-negative basis} as 
an orthonormal basis $\{\rho(j):j\in\mbb{Z}_d\}$ such that for all $j\in\mbb{Z
}_d$ and $\lambda\in\Lambda$,
\begin{subequations}\begin{align}
\mu_{\rho(j)}(\lambda)&\ge 0	\,,\\
\xi_{\rho(j)}(\lambda)&\in[0,1]	\,.
\end{align}\end{subequations}
That is, each state in a non-negative basis is a non-negative state and the 
projective measurement corresponding to the non-negative basis is a 
non-negative measurement. Considering non-negative bases as corresponding 
to both a basis of non-negative states and a non-negative projective 
measurement is the fundamental tool that we will exploit to obtain 
the results of this paper. Note that when we consider non-negative bases,
we only consider orthonormal bases of pure states.

We are particularly interested in subtheories of quantum mechanics that 
contain non-negative bases and a nontrivial group of transformations that 
permute the non-negative bases. For qubits, we will establish upper bounds on 
the number of non-negative bases and some relations any non-negative bases 
must satisfy. We will also completely classify the possible sets of 
non-negative bases that are closed under a nontrivial unitary group. 
For qudits, we will only undertake the first task (i.e., establish upper 
bounds on the number of non-negative bases and relations between any 
non-negative bases). 

We begin by establishing some properties that any quasiprobability 
representation must satisfy. These properties were all proven in 
Ref.~\cite{Spekkens2005}, but are included here for completeness. 
We first prove that the supports $\supps{\rho(j)}$ of the states in a 
non-negative basis $\{\rho(j):j\in\mbb{Z}_d\}$ must be disjoint. 
Furthermore, for any $\lambda$ that is compatible with one of the 
basis states, the indicator functions $\xi_{\rho(j)}(\lambda)$ are 
outcome deterministic and correspond to answering the question 
``is $\lambda$ compatible with $\rho(j)$?''. 

\begin{lem}\label{lem:deterministic_indicator}
Let $\{\rho(j):j\in\mbb{Z}_d\}$ be a non-negative basis of $\mc{H}_d$ in a 
given quasiprobability representation. Then the supports $\{\supps{\rho(j)}:j
\in \mbb{Z}_d\}$ are disjoint, and for all $j,k\in\mbb{Z}_d$ we have
\begin{align}\label{eq:deterministic_indicator}
\xi_{\rho(j)}(\lambda) = \delta_{j,k}\ \forall\lambda\in\supps{\rho(k)}	\,.
\end{align}
\end{lem}

\begin{proof}
Let $\{\rho(j):j\in\mbb{Z}_d\}$ be a non-negative basis. As non-negative 
bases are orthonormal by definition, $\trs{\rho(j)\rho(k)} = \delta_{j,k}$. 
Consequently, for the preparation $\rho(j)$ followed by a measurement in this 
basis, the Born rule gives
\begin{equation}
\int_{\supps{\rho(j)}}d\lambda\, \mu_{\rho(j)}(\lambda)\xi_{\rho(k)}(\lambda)
= \delta_{j,k}	\,.
\end{equation}
for all $j,k\in\mbb{Z}_d$. As $\mu_{\rho(j)}(\lambda)\geq 0$ and is 
normalized, and $\xi_{\rho(k)}(\lambda)\in[0,1]$, the only solution\footnote{
Here and elsewhere, except possibly on a set of measure zero.} is as in 
Eq.~\eqref{eq:deterministic_indicator}.
\end{proof}

Quasiprobability representations are convex-linear (i.e., they are affine 
maps, preserving convex combinations of states and 
measurements)~\cite{Ferrie2009}. We now show that,
for a fixed quasiprobability representation with at least one 
non-negative basis, convex-linearity implies that any point $\lambda$
is either compatible with exactly one element of each non-negative basis, 
or is incompatible with all elements of every non-negative basis. 
Any points $\lambda$ that are incompatible with all elements of every
non-negative basis are irrelevant to the description of the non-negative 
subtheory, so we define the space $\Lambda_{*}\subseteq\Lambda$ 
by deleting such points. Furthermore, we also show that there exists a 
unique function $q:\Lambda_{*}\to \mathbb{R}^{+}$ such that any element 
that assigns nonzero probability to $\lambda$ assigns probability 
$q(\lambda)$. This unique value is given by 
$q(\lambda) = d \cdot \mu_{\frac{1}{d}\unit}(\lambda)$, that is, $d$ 
times the probability that the maximally-mixed state assigns to the 
ontic state $\lambda$.

\begin{lem}\label{lem:sum_prob}
For any quasiprobability representation of $\mc{H}_d$ in which 
there is at least one non-negative basis, there exists a unique function 
$q:\Lambda\to\mathbb{R}^{+}\cup\{0\}$ such that for every 
non-negative basis $\{\rho(j)\}$, $\Lambda$ can be 
partitioned into $d+1$ disjoint regions 
$\{\supps{\rho(j)},\Lambda_0:j\in\mbb{Z}_d\}$ such that
\begin{align}
\mu_{\rho(j)}(\lambda) &= \begin{cases} \delta_{j,k}q(\lambda) & 
\forall\lambda\in\supps{\rho(k)}\,,	\\ 0 & \forall\lambda\in\Lambda_0	\,,\end{cases}
\end{align}
for all $j,k\in\mbb{Z}_d$, where $\Lambda_0=\{\lambda:q(\lambda)=0\}$.
\end{lem}

\begin{proof}
Let $\{\rho(j)\}$ be an orthonormal basis of $\mc{H}_d$. Then the maximally 
mixed state can be written as
\begin{align}
\frac{1}{d}\unit = \frac{1}{d}\sum_j \rho(j)	\,.
\end{align}
As the quasiprobability representation is convex-linear,
\begin{align}\label{eq:convex-linear_distribution}
\mu_{\frac{1}{d}\unit}(\lambda) = \frac{1}{d}\sum_j \mu_{\rho(j)}(\lambda)
\end{align}
for all $\lambda\in\Lambda$. Set $q(\lambda) = d \cdot \mu_{\frac{1}{d}\unit}(
\lambda)$.

For any non-negative basis $\{\rho(j)\}$, each term in the right-hand side of 
Eq.~\eqref{eq:convex-linear_distribution} is non-negative by definition and 
so $q(\lambda)$ must be non-negative. Furthermore, any point $\lambda\in\Lambda$ 
is compatible with at most one element of any non-negative 
basis by Lemma~\ref{lem:deterministic_indicator}.
\end{proof}

We now show that two quantum states that are elements of non-negative bases 
have disjoint supports if and only if they are orthogonal quantum states.

\begin{lem}\label{lem:common_support}
Let $\rho$ and $\omega$ be elements of (possibly identical) non-negative 
bases in a quasiprobability distribution. Then $\supp{\rho}\cap\supp{\omega}=
\emptyset$ if and only if $\rho$ and $\omega$ are orthogonal.
\end{lem}

\begin{proof}
Let $\rho$ and $\omega$ be elements of (possibly identical) non-negative 
bases. By Lemma~\ref{lem:deterministic_indicator} and \ref{lem:sum_prob}, $\xi
_{\rho}(\lambda)$ is nonzero only for $\lambda$ such that $\mu_{\rho}(\lambda)>0$ 
or $q(\lambda)=0$. By Lemma~\ref{lem:sum_prob}, $\mu_{\omega}(\lambda)$ 
can only be nonzero for $\lambda$ such that $q(\lambda)>0$. Therefore, for 
\begin{align}
\int_{\Lambda}d\lambda\, \mu_{\omega}(\lambda)\xi_{\rho}(\lambda) &= \tr{\rho 
\omega}
\end{align}
to hold, there must exist some $\lambda$ such that $\mu_{\rho}(\lambda)\neq 0$
and $\mu_{\omega}(\lambda)\neq 0$ if and only if $\rho$ and $\omega$ are not 
orthogonal.
\end{proof}

Lemmas~\ref{lem:deterministic_indicator}--\ref{lem:common_support} provide 
the basic mathematical tools that we use to provide an upper bound on the 
number of non-negative bases of $\mc{H}_d$ in any quasiprobability 
representation of quantum mechanics. Any quasiprobability representation over 
a space $\Lambda$ which is non-negative when restricted to a subset of bases 
of $\mc{H}_d$ must satisfy the following: the representation of a 
non-negative basis $\{\rho(j)\}$ of $\mc{H}_d$ must correspond to a 
partitioning of $\Lambda$ into $d$ disjoint regions $\{\supps{\rho(j)}\}$ 
such that $\rho(j)$ assigns nonzero probability for all 
$\lambda\in\Lambda_j$ and 0 probability elsewhere. The measurement in 
the non-negative basis then corresponds to determining which region
the ontic state is in.

\section{Quasi-probability representations of qubits}\label{sec:qubits}

We now consider the simplest case, namely, quasiprobability representations 
of qubits (i.e., two-dimensional quantum systems). We will first establish 
upper bounds on the number of non-negative bases for a single qubit and 
necessary relations between sets of non-negative bases. We will then construct 
a complete characterization of all the sets of bases that are closed under a 
nontrivial unitary group and are simultaneously non-negative in some 
quasiprobability representation.

The qubit case allows a geometrical approach, as any qubit state $\rho\in\mc{B
}(\mc{H}_2)$ can be written as
\begin{align}
\rho = \frac{1}{2}\left(\unit + \vec{r}(\rho)\cdot\vec{\sigma}\right)	\,,
\end{align}
where $\vec{r}(\rho)\in\mathbb{R}^3$ is the Bloch vector corresponding to 
$\rho$ and $\vec{\sigma}$ is the vector $\left(X,Y,Z\right)$ of Pauli 
matrices. Orthogonal quantum states then correspond to antipodal Bloch 
vectors, i.e.,
\begin{align}
\tr{\rho \omega} = 0 \Leftrightarrow \vec{r}(\rho)=-\vec{r}(\omega)	\,.
\end{align}
Another useful feature of the qubit case is that any pure single qubit state 
can be uniquely extended to an orthonormal basis of a qubit. Therefore we can 
represent a basis by a single Bloch vector $\vec{r}(j)$ and denote the basis 
elements by
\begin{align}
\{\rho(j,\gamma) := \frac{1}{2}\left[\unit + \gamma \vec{r}(j)\cdot \vec{\sigma}\right]:\gamma=\pm\}	\,.
\end{align}

\subsection{Maximum number of non-negative qubit bases}\label{sec:qubit_bound}

We begin by showing that a quasiprobability representation cannot be 
non-negative for more than two bases in any plane of the Bloch sphere. That is, 
for any quasiprobability distribution, the Bloch vectors corresponding to any 
three distinct non-negative bases (if three such bases exist) must be 
linearly independent. We will then prove that if four bases are non-negative 
in some quasiprobability representation of a qubit, then they correspond to 
the vertices of a right cuboid. 

\begin{thm}\label{thm:qubit_planar}
In any quasiprobability representation of a qubit, there are at most two 
non-negative bases in any plane of the Bloch sphere.
\end{thm}

\begin{proof}

The proof is by contradiction. Let $\{\rho(j,\gamma):j=1,2,3,\gamma=\pm\}$ be 
three distinct coplanar bases that are non-negative in some quasi-probability 
representation. We rotate the Bloch sphere so that the bases are in the $xz$ 
plane of the Bloch sphere and label the bases as shown in 
Fig.~\ref{fig:co-planar_bases}.

As the three vectors $\{\vec{r}(j):j=1,2,3\}$ are linearly dependent, there 
exists a vector $\vec{r}(\rho)$ and $a,b\in[0,1]$ such that
\begin{align}\label{eq:coplanar_bases_vector}
\vec{r}(\rho) &= a\vec{r}(2) + (1-a)[-\vec{r}(2)]	\nonumber\\
&= b \vec{r}(1) + (1-b)\vec{r}(3) \,.
\end{align}
Therefore there exists a mixed state $\rho$ such that
\begin{align}\label{eq:coplanar_bases_decomposition}
\rho &= a\rho(2,+) + (1-a)\rho(2,-)	\nonumber\\
&= b \rho(1,+) + (1-b) \rho(3,+)	\,.
\end{align}

As $\mu$ is convex-linear, we have
\begin{align}\label{eq:qubit_mixed}
\mu_{\rho}(\lambda) &= a\mu_{\rho(2,+)}(\lambda) + (1-a)\mu_{\rho(2,-)}(
\lambda)	\nonumber\\
&= b\mu_{\rho(1,+)}(\lambda) + (1-b)\mu_{\rho(3,+)}(\lambda)	\,.
\end{align}
for all $\lambda\in\Lambda$. 

As the three bases are distinct, neither $a$ or $b$ can be $0$ or $1$. As 
$\rho(1,+)$ and $\rho(3,+)$ are not orthogonal, there exists 
$\lambda'\in\supps{\rho(1,+)}\cap\supps{\rho(3,+)}$ by 
Lemma~\ref{lem:common_support}. Then by Lemma~\ref{lem:sum_prob}, the first 
line of Eq.~\eqref{eq:qubit_mixed} gives $\mu_{\rho}(\lambda') < q(\lambda')$, 
while the second gives $\mu_{\rho}(\lambda') = q(\lambda')$. 
\end{proof}

\begin{figure}
	\centering
	\includegraphics[width=0.75\linewidth]{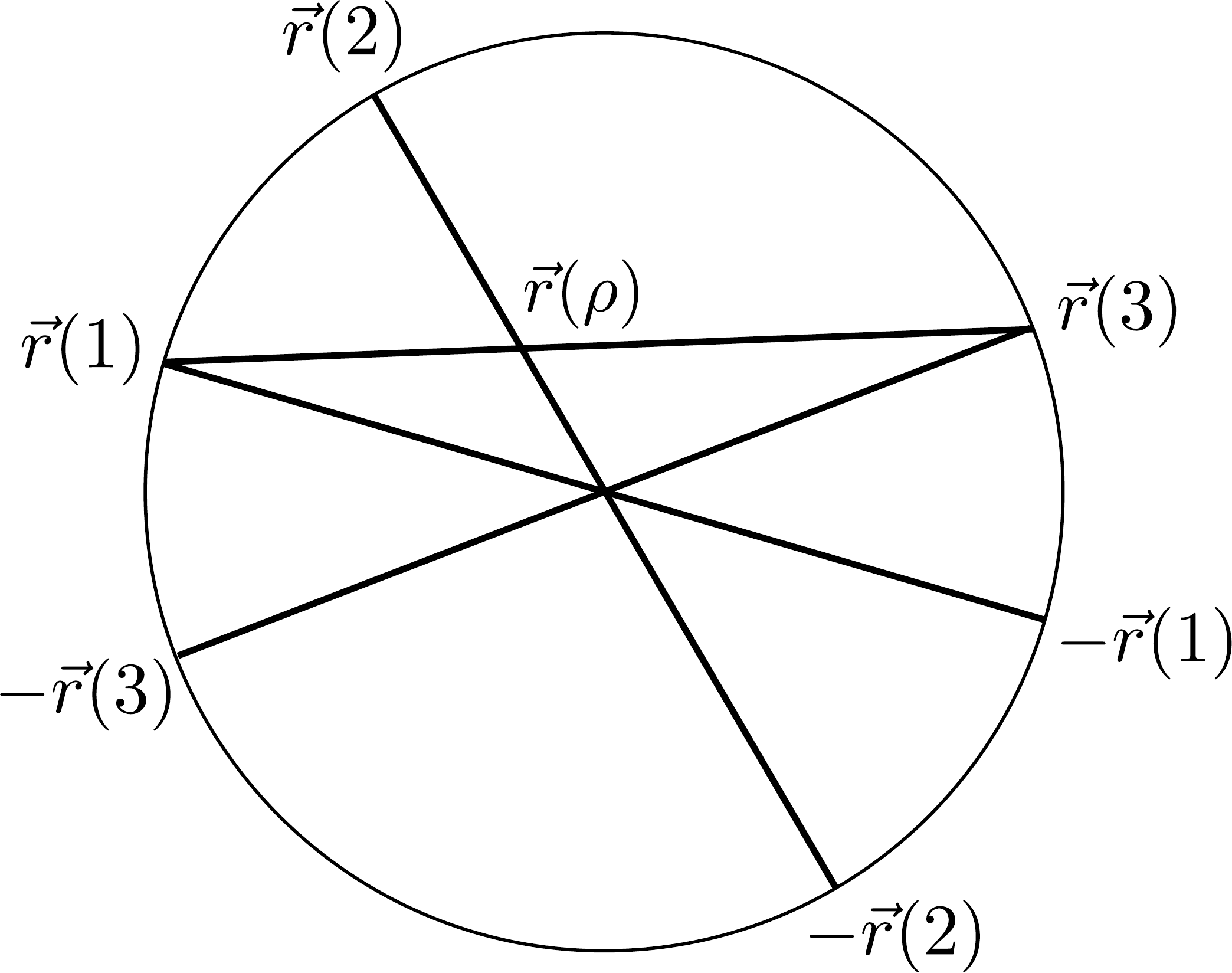}
	\caption{Diagrammatic representation of the decomposition of the vector $
\vec{r}(\rho)$ in Eq.~\eqref{eq:coplanar_bases_vector} in terms of $\{\pm\vec{
r}(2)\}$ and $\{\vec{r}(1),\vec{r}(3)\}$.}\label{fig:co-planar_bases}
\end{figure}

The proof that no three non-negative bases can be coplanar in the Bloch sphere 
rests on the existence of the vector in Eq.~\eqref{eq:coplanar_bases_vector}. 
This vector must always exist as otherwise there would be three linearly 
independent vectors in a two-dimensional plane, which would be a contradiction.

Initially, one might expect that this proof could be generalized to rule out 
four non-negative bases in the full Bloch sphere (i.e., the existence of four 
non-negative bases would correspond to four linearly independent Bloch 
vectors, contradicting the dimensionality of a sphere). This intuition almost 
always holds. However, as we now prove, there is a family of exceptions with 
a high degree of symmetry, namely, if the four bases correspond to the 
vertices of a right cuboid (i.e., a solid with rectangular faces). Therefore 
any set of four non-negative bases of a qubit must correspond to the vertices 
of a right cuboid, which also implies that there exists no quasiprobability 
representation of a qubit with five or more non-negative bases.

\begin{thm}\label{thm:qubit_new}
If there are four non-negative bases in a quasiprobability representation of 
a qubit, then the Bloch vectors corresponding to these four bases must 
correspond to the vertices of a right cuboid.
\end{thm}

\begin{proof}

The proof is by contradiction. Let $\{\rho(j,\gamma):j=1,2,3,\gamma=\pm,4\}$ 
be four bases that are non-negative in a quasiprobability representation of a 
qubit. Negating the vectors $\vec{r}_j$ as necessary (which corresponds to 
relabeling the basis states), there exists a vector $\vec{r}(\rho)$ and $a,s_j
\in[0,1]$ such that
\begin{align}\label{eq:combination}
\vec{r}(\rho) &= a\vec{r}(4) + (1-a)[-\vec{r}(4)] \nonumber\\
&=	\sum_{j=1}^3 s_j \vec{r}(j)	\,,
\end{align}
where $\sum_j s_j = 1$, as otherwise the four vectors $\{\vec{r}(j):j=1,2,3,4
\}$ would be linearly independent.

Therefore there exists a state $\rho$ such that
\begin{align}\label{eq:decomposition}
\rho &=a\rho(4,+) + (1-a)\rho(4,-) \nonumber\\
&= \sum_{j=1}^3 s_j\rho(j,+)	\,.
\end{align}
As $\mu$ is convex-linear,
\begin{align}\label{eq:qubit_mixed2}
\mu_{\rho}(\lambda) &= 	a\mu_{\rho(4,+)}(\lambda) + (1-a)\mu_{\rho(4,-)}(
\lambda) \nonumber\\
&= \sum_{j=1}^3 s_j\mu_{\rho(j,+)}(\lambda) \,,
\end{align}
for all $\lambda\in\Lambda$. 

We now use the properties of quasiprobability representations to restrict the 
values of the coefficients and then show that the only solutions to 
Eq.~\eqref{eq:combination} with the appropriate coefficients correspond to the vertices 
of a right cuboid.

As the four bases are distinct and cannot be coplanar by 
Theorem~\ref{thm:qubit_planar}, none of the coefficients 
can be $0$ or $1$. By Lemma~\ref{lem:sum_prob}, the first 
line of Eq.~\eqref{eq:qubit_mixed2} implies 
$\mu_{\rho}(\lambda) \in\{a q(\lambda), (1-a) q(\lambda)\}$ 
for all $\lambda\in\Lambda_{*}$. 

As $\rho(1,-)$ and $\rho(2,-)$ are not orthogonal, there exists $\lambda_3
\in\supps{\rho(1,-)}\cap\supps{\rho(2,-)}$ by Lemma~\ref{lem:common_support}. 
Lemma~\ref{lem:sum_prob} then implies $s_3 \in \{a,1-a\}$. Interchanging the 
roles of the first three bases and iterating the above argument implies $s_j
\in\{a,1-a\}$ for $j=1,2,3$. As $\sum_j s_j = 1$ and $a\notin\{0,1\}$, we 
have $s_1=s_2=s_3=\tfrac{1}{3}$ and $a\in\{\tfrac{1}{3},\tfrac{2}{3}\}$.

Substituting these coefficients back into Eq.~\eqref{eq:combination} and 
swapping the sign of $\vec{r}(4)$ if $a=\tfrac{2}{3}$ gives
\begin{align}\label{eq:combination_coefficients}
\vec{r}(4) = \vec{r}(1) + \vec{r}(2) + \vec{r}(3)	\,.
\end{align}
We can rotate the Bloch sphere so that
\begin{align}\label{eq:rotated_vectors}
\vec{r}_1 &= \left(-\sin\theta,0,\cos\theta\right)	\nonumber\\
\vec{r}_2 &= \left(\sin\theta,0,\cos\theta\right)
\end{align}
for some $\theta\in(0,\tfrac{\pi}{2})$. Substituting 
Eq.~\eqref{eq:rotated_vectors} into 
Eq.~\eqref{eq:combination_coefficients} gives 
$\vec{r}_x(3)= \vec{r}_x(4)$ and $\vec{r}_y(3)= \vec{r}_y(4)$. 
As $\vec{r}(3)$ and $\vec{r}(4)$ are distinct unit vectors 
(as they are pure states belonging to distinct bases by assumption),
$-\vec{r}_z(3) = \vec{r}_z(4) =\cos\theta$. Therefore the vectors
$\{\vec{r}(1),\vec{r}(2),-\vec{r}(3),\vec{r}(4)\}$ correspond to 
the vertices of a rectangle in the plane $z=\cos\theta$, so the vectors 
$\{\pm\vec{r}(j):j=1,2,3,4\}$ correspond to the vertices of a right cuboid.
\end{proof}

Theorems~\ref{thm:qubit_planar} and \ref{thm:qubit_new} show that in any 
quasiprobability representation of a qubit, any three non-negative bases must 
have linearly independent Bloch vectors and any four non-negative bases must 
correspond to the vertices of a right cuboid. We will find that these 
Theorems alone do not completely characterize the sets of bases that can 
simultaneously be non-negative in a quasiprobability representation. In 
particular, there are further restrictions on the sets of three bases that 
can simultaneously be non-negative. 

\subsection{Sets of non-negative bases that allow nontrivial transformations}
\label{sec:examples}

We now further restrict our consideration to non-negative subtheories of a 
qubit that contain at least one nontrivial unitary transformation. By 
nontrivial, we mean that the transformation is not a phase-multiple of the 
identity, which, for a set of two or more bases, implies that a nontrivial 
unitary gives a nontrivial permutation of the states in the subtheory. Such 
non-negative subtheories are analogous to the subtheory of single qubit 
stabilizers states, measurements in the corresponding bases, and single qubit 
Clifford transformations.

The set of Bloch vectors corresponding to a set of non-negative bases can be 
regarded as a set of pairs of antipodal points on the surface of a (Bloch) 
sphere. If there are transformations that permute the elements of these 
bases, then they correspond to elements of the point group of the set of 
points on the surface of the sphere, which will be a discrete subgroup of $O(3)$ 
(except in the special case wherein there is only one non-negative basis). 
For the transformations to have a unitary representation, they must have 
determinant +1 [i.e., be elements of $SO(3)$]. As we are interested in 
subtheories of quantum mechanics, we restrict to elements of
$SO(3)$~\footnote{We only consider elements of $SO(3)$ rather than $O(3)$ 
because we aim to describe quantum mechanics, or a subtheory thereof, for 
which only unitary and not anti-unitary transformations are allowed. However, 
it is possible to extend quantum mechanics to include anti-unitary 
transformations, and if one wished to consider ontological models describing
such an extended theory, then all transformations in $O(3)$ (including 
reflections) can be included. Because we are considering subtheories with 
only discrete evolution, the restriction to unitary dynamics is not necessarily 
well-motivated. Specifically, because the evolution in a discrete theory is 
not required to be continuously deformable to the identity, such an extended 
model may be reasonable.}.

The finite subgroups of $SO(3)$ have been completely classified. The only 
possible groups that can permute a set of non-negative bases 
are~\cite{Tinkham2003}: 
\begin{itemize}
\item D$_{\infty}$, the group of rotations about an axis and $\pi$-flips 
around some orthogonal axis;
\item C$_2 \simeq\mbb{Z}_2$, the cyclic group of two elements;
\item D$_2$, the symmetry group of a rectangle;
\item D$_3$, the symmetry group of an equilateral triangle;
\item D$_4$, the symmetry group of a square; and
\item O$_h$, the octahedral group,
\end{itemize}
where we have used the results of the previous section to ignore groups that 
are not the point group of a set of one, two or three linearly independent 
vectors and their negatives, or the vertices of a right cuboid.

We now completely characterize the sets of bases that are non-negative in 
some quasiprobability representation of a qubit and in addition are closed 
under a nontrivial unitary group. We are explicit with our construction for 
the cases of three and four non-negative bases; quasiprobability 
representations that are non-negative for one or two bases appear as subsets 
of those for three non-negative bases, and therefore we do not consider them 
separately.

\subsubsection*{One non-negative basis}

If there is a single non-negative basis (i.e., $\{\ket{0},\ket{1}\}$ up to 
unitary equivalence), then the group of transformations that permute the 
elements of the non-negative basis is the continuous (Lie) group D$_{\infty}$. 
We can therefore find a quasiprobability representation over a space 
consisting of two points---a classical bit---for which this basis is 
represented by non-negative probability distributions.

The transformations of the non-negative basis is generated by the group of $U(1)$
rotations about the $z$-axis (which leaves the basis states invariant) 
and $X$, which corresponds to a bit-flip.

\subsubsection*{Two non-negative bases}

If there are two non-negative bases, then the point group of the four 
vertices $\{\pm\vec{r}(1),\pm\vec{r}(2)\}$ is generically a representation of 
D$_2$. Expressing the Bloch vectors as
\begin{align}
\vec{r}(1) &= \left(\sin\theta,0,\cos\theta\right)	\,,\nonumber\\
\vec{r}(2) &= \left(-\sin\theta,0,\cos\theta\right)	\,,
\end{align}
the symmetry group is the group of Pauli matrices. The two non-negative bases 
and the generators of the symmetry group are illustrated in Fig.~\ref{fig:D2}. 
For the special case defined by $\vec{r}(1)\cdot\vec{r}(2)=0$, the four 
vertices correspond to the corners of a square and so the point group is D$_4$, 
which for the above vectors is generated by $X$ and $Y^{\frac{1}{2}}$.

\begin{figure}[t!]
	\includegraphics[width=\linewidth]{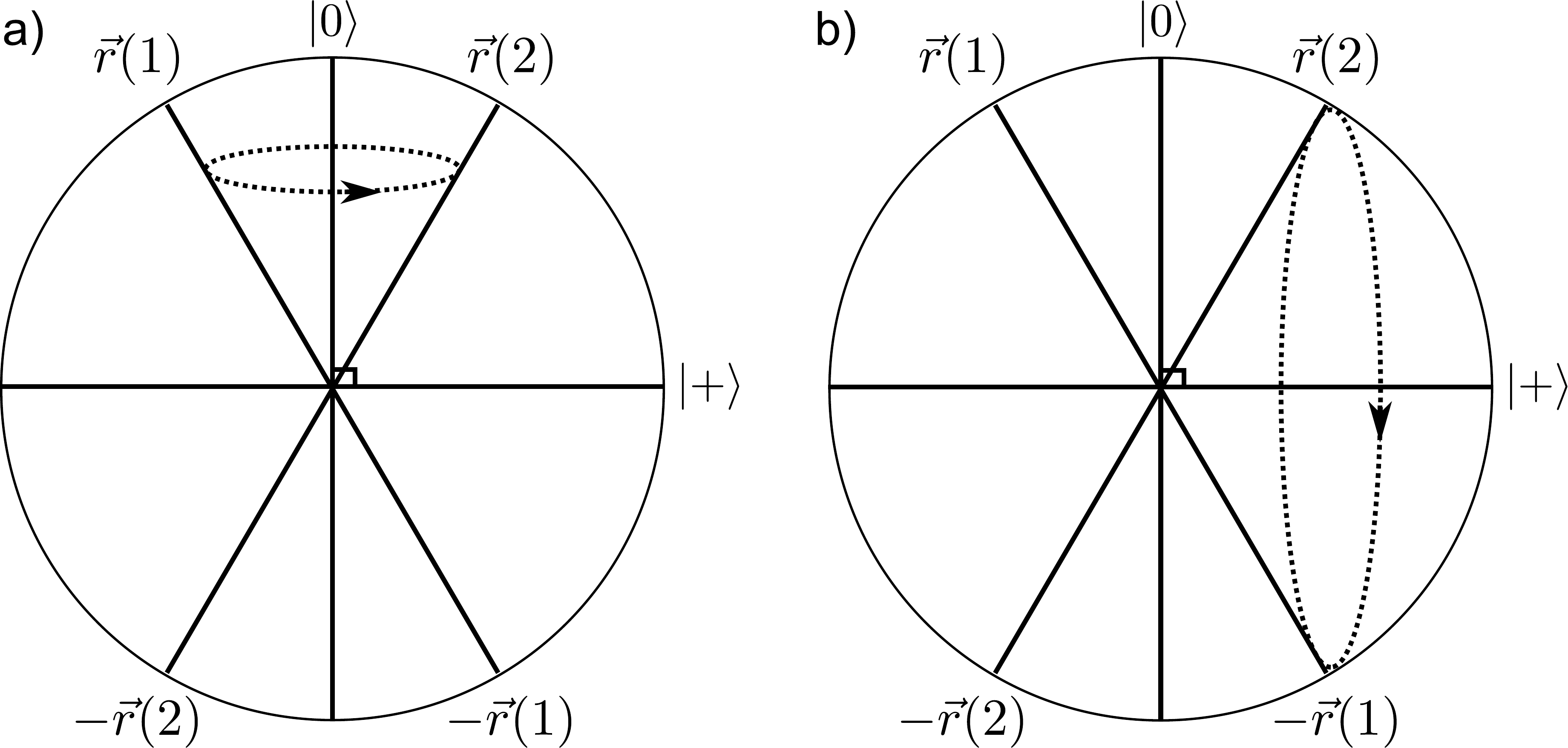}
	\caption{Bloch vector representation of the transformations $Z$ and $X$ 
that map: a) $\vec{r}(1)\leftrightarrow\vec{r}(2)$ and $-\vec{r}(1)
\leftrightarrow-\vec{r}(2)$ and b) $\vec{r}(1)\leftrightarrow-\vec{r}(2)$ and 
$-\vec{r}(1)\leftrightarrow\vec{r}(2)$, respectively.}\label{fig:D2}
\end{figure}

\subsubsection*{Three non-negative bases}

There are two inequivalent families of sets of three bases with nontrivial 
point groups, D$_3$ and $\mathbb{Z}_2$, respectively. Both families have 
special cases for which the three bases have a higher degree of symmetry. In 
particular, both families include a unique special case in common for which 
the point group is the octahedral group. For this special case, we recover 
the standard single qubit stabilizer states, whose point group is the single 
qubit Clifford group.

We wish to construct quasiprobability representations over some ontic state 
space $\Lambda$, in which the set of three bases $\{\rho(j,\gamma):j=1,2,3,
\gamma=\pm\}$ are non-negative. By Lemma~\ref{lem:sum_prob}, any non-negative 
basis corresponds to a bipartition of the ontic space $\Lambda$ such that 
each element of the basis only has support on one of the partitions. 
Furthermore, the probability that any non-negative state $\rho$ assigns to an 
ontic state $\lambda\in\supp{\rho}$ must be $q(\lambda)$, which is double the 
probability the maximally mixed state assigns to $\lambda$. Therefore, with 
three non-negative bases, $\Lambda$ can be partitioned into 8 regions 
corresponding to the 8 combinations of basis states that are non-negative in 
that region. Without loss of generality, we can integrate over these regions 
and, for convenience, we parametrize the ontic state space as the 8 points
\begin{align}
\Lambda = \{(\epsilon,a):\epsilon=\pm,a\in\mathbb{Z}_4\} \,,
\end{align}
and henceforth we write $q(\lambda) = q(\epsilon,a)$.

Our (arbitrary) choices of six distributions over these eight ontic states 
are listed in Tab.~\ref{tab:supports}.

\begin{table}[t!]
\begin{center}
\begin{tabular}{c|c}
\hline\hline
Quantum state & Support \\\hline
$\rho(1,+)$ & $(+,0),\ (+,1),\ ({-},2),\ ({-},3)$ \\
$\rho(1,-)$ & $({-},0),\ ({-},1),\ (+,2),\ (+,3)$ \\
$\rho(2,+)$ & $(+,0),\ ({-},1),\ (+,2),\ ({-},3)$ \\
$\rho(2,-)$ & $({-},0),\ (+,1),\ ({-},2),\ (+,3)$ \\
$\rho(3,+)$ & $(+,0),\ ({-},1),\ ({-},2),\ (+,3)$ \\
$\rho(3,-)$ & $({-},0),\ (+,1),\ (+,2),\ ({-},3)$ \\
\hline\hline
\end{tabular}
\end{center}
\caption{\label{tab:supports} List of the six possible supports over the set 
of ontic states $\{(\epsilon,a):\epsilon=\pm,a\in\mathbb{Z}_4\}$ for the 
elements of three non-negative bases $\{\rho(j,\gamma):j=1,2,3,\gamma=\pm\}$.}
\end{table}

The distribution $q(\epsilon,a)$ will have to satisfy certain constraints in 
order to reproduce the correct quantum mechanical predictions. The 
distribution must be normalized over the support of each non-negative state, 
i.e.,
\begin{align}\label{eq:normalizations}
\sum_{(\epsilon,a)\in\supps{\rho(j,\gamma)}} q(\epsilon,a) &= 1
\end{align}
for $\gamma=\pm$ and $j=1,2,3$. These six equations give only four 
independent constraints.

In addition, in order for the quasiprobability representation to reproduce 
the quantum probabilities of preparing a system in one non-negative basis and 
then measuring in another non-negative basis, we also require
\begin{align}\label{eq:3_bases_overlaps}
\sum_{(\epsilon,a)\in\supps{\rho(j_1,\gamma_1)}\cap\supps{\rho(j_2,\gamma_2)}}
q(\epsilon,a) & = \frac{1}{2}\left[1 + \gamma_1\gamma_2\vec{r}(j_1)\cdot\vec{r
}(j_2)\right]	\,,
\end{align}
for all $j_1 \neq j_2$, which are obtained by substituting 
Eq.~\eqref{eq:deterministic_indicator} into Eq.~\eqref{eq:ontic_probability}.
Only three of these equations (e.g., $j_1> j_2$ and $\gamma_1=\gamma_2 = +$)
give independent constraints. Therefore, for the cases we consider,
we obtain four independent constraints from Eq.~\eqref{eq:normalizations}
and three further independent constraints from 
Eq.~\eqref{eq:3_bases_overlaps}. With $q(\epsilon,a)$ defined on 
eight points, the seven independent constraints ensure there will
be at most a one-parameter family of quasiprobability representations 
for which the desired sets of states are all non-negative.

To find a quasiprobability representation of all states and measurements of a 
qubit that is non-negative for the above bases, we need to find operators 
$F(\epsilon,a)$ and $G(\epsilon,a)$ satisfying Eqs.~\eqref{eq:quasi_F}, 
\eqref{eq:quasi_G} and \eqref{eq:ontic_probability}. As these operators 
should give a non-negative distribution for our chosen bases, we require
\begin{align}
\mu_{\rho(j,\gamma)}(\epsilon,a) &= \trs{F(\epsilon,a)\rho(j,\gamma)} 
\nonumber\\
&= \begin{cases} q(\epsilon,a) & \text{if } (\epsilon,a)\in\supps{\rho(j,
\gamma)}\,,	\\ 0 & \text{otherwise}	\,,\end{cases}
\end{align}
for all $\epsilon,\gamma=\pm$, $j=1,2,3$ and $a=1,2,3,4$. To find such $F(
\epsilon,a)$ and $G(\epsilon,a)$, we find vectors $\vec{d}(a)$ such that
\begin{align}\label{eq:ontic_to_quasi}
\vec{d}(a)\cdot\vec{r}(j) &= \begin{cases} 1 & \text{if } (+,a)\in\supps{\rho(
j,+)}\,,	\\ -1 & \text{otherwise}	\,.\end{cases}
\end{align}
The operators
\begin{align}
F(\epsilon,a) &= \frac{q(\epsilon,a)}{2}\left(\unit + \epsilon\vec{d}(a)
\cdot\vec{\sigma}\right)	\nonumber\\
G(\epsilon,a) &= \frac{1}{2}\left(\unit + \epsilon\vec{d}(a)\cdot\vec{\sigma}
\right)
\end{align}
then define a quasiprobability representation that is non-negative for the 
three bases. Note that the operators $F(\epsilon,a)$ and $G(\epsilon,a)$ are 
identical up to normalization. In what follows, we will only explicitly 
present the operators $F(\epsilon,a)$.

We now turn separately to each of the inequivalent families and construct 
quasiprobability representations that are non-negative for each of the 
families. As both families include the single qubit stabilizer states as a 
special case, we will discuss this case after covering the general case for 
each family.

\medskip\paragraph*{\textbf{Case 1.}} Up to an overall unitary, the first 
family is the one-parameter family of three bases illustrated in 
Fig.~\ref{fig:D3} with
\begin{align}\label{eq:D3_bases}
\vec{r}(1) &= \left(\sin\theta,0,\cos\theta\right)	\,,\nonumber\\
\vec{r}(2) &= \left(-\tfrac{1}{2}\sin\theta,\tfrac{\sqrt{3}}{2}\sin\theta,
\cos\theta\right)	\,,\nonumber\\
\vec{r}(3) &= \left(-\tfrac{1}{2}\sin\theta,-\tfrac{\sqrt{3}}{2}\sin\theta,
\cos\theta\right)	\,,
\end{align}
for $\theta\in(0,\pi)$. The corresponding point group is D$_3$, which is 
generated by a $\frac{2\pi}{3}$ rotation about the $z$-axis (which we denote 
by $\Gamma$) and a $\pi$ rotation about the $y$-axis (which we denote by $\Pi$). 
These transformations do not commute as D$_3$ is a non-abelian group.

\begin{figure}[t!]
\centering
	\includegraphics[width=0.75\linewidth]{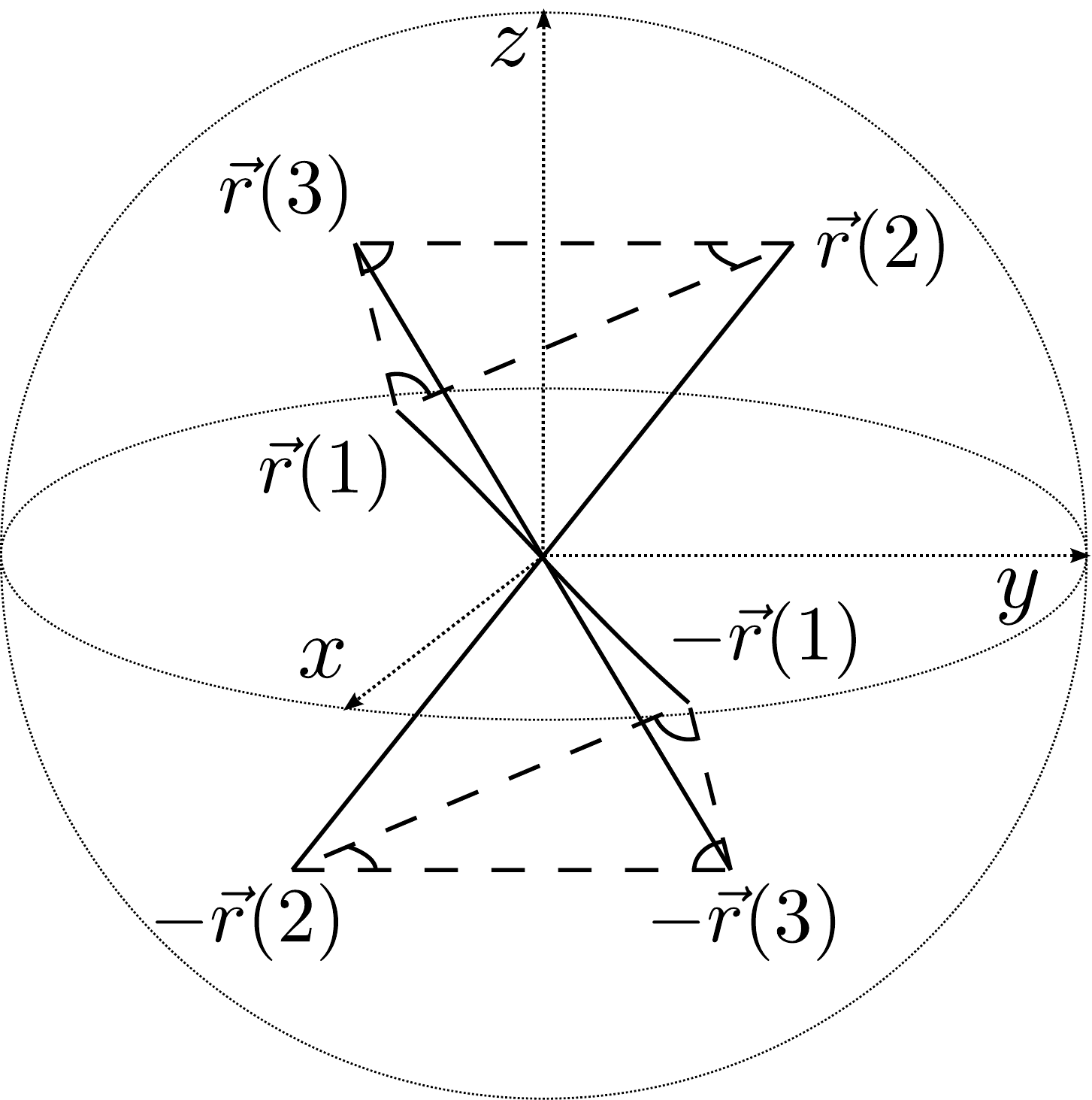}
	\caption{Illustration of the Bloch vectors in Eq.~\eqref{eq:D3_bases}, 
whose point group is D$_3$.}\label{fig:D3}
\end{figure}

To define a quasiprobability representation for which these three bases are 
non-negative, we need to find a distribution $q(\epsilon,a)$ satisfying 
Eq.~\eqref{eq:normalizations} and \eqref{eq:3_bases_overlaps}. The only 
such distributions are
\begin{align}
	q(+, 0) &= q_0	\,,\nonumber\\
	q(+, a) &= \frac{3}{2}\sin^2\theta - 1 + q_0\,,\ a=1,2,3	\,,\nonumber\\
	q(-, 0) &= 2 - q_0 - \frac{9}{4}\sin^2\theta	\,,\nonumber\\
	q(-, a) &= 1 - q_0 - \frac{3}{4}\sin^2\theta \,,\ a=1,2,3\,,
\end{align}
where $q_0\in[0,1]$ is a free parameter. The requirement that all 
probabilities should be in the interval $[0,1]$ implies
\begin{align}
0 \leq q_0 \leq 2-\frac{9}{4}\sin^2\theta	\,,
\end{align}
which can only be satisfied when $\sin^2\theta \leq \frac{8}{9}$. Therefore a 
quasiprobability representation for which these bases are non-negative can 
only be defined if $\sin^2\theta \leq \frac{8}{9}$. This constraint is not a 
consequence of either Theorem~\ref{thm:qubit_planar} or \ref{thm:qubit_new} 
and so provides an additional constraint on the set of non-negative bases.

In an ontological model, transformations are fundamentally transformations of 
ontic states, not epistemic states (i.e., states of knowledge). For an 
ontological model of a subtheory of quantum mechanics in which (pure) quantum 
states are epistemic states, this means that unitary transformations must 
\textit{supervene on} transformations of ontic states (i.e., must be a 
consequence of some underlying transformation of the ontic states). In the 
models we consider, we will always be able to assume that transformations of 
ontic states are deterministic (i.e., correspond to a permutation of the 
ontic state space). Conversely, some (but not all) permutations of ontic 
states in these models \textit{effect} a unitary transformation.

The permutations of ontic states that can effect the rotations $\Gamma$ and 
$\pi$ are illustrated in Fig.~\ref{fig:D3_transformations}~(g) and (h), 
respectively. The permutation in Fig.~\ref{fig:D3_transformations}~(h) only 
permutes the probability distributions for non-negative states if $q(+,0)=q(-,
0)$, which fixes
\begin{align}\label{eq:D3d_probabilities}
	q(\epsilon,a) &= q_0 = 1 - \frac{9}{8}\sin^2\theta	\,,\nonumber\\
	q(\epsilon,a) &= \frac{3}{8}\sin^2\theta := q_1	\,,\ a=1,2,3	\,,
\end{align}
for $\epsilon=\pm$, as illustrated in Fig.~\ref{fig:D3_transformations}~(a)-(f). 
Note that the transformation $\Gamma$ always supervenes on the permutation 
in Fig.~\ref{fig:D3_transformations}~(g), even if $q(+,0)\neq q(-,0)$.

\begin{figure}
\centering
	\includegraphics[width=0.95\linewidth]{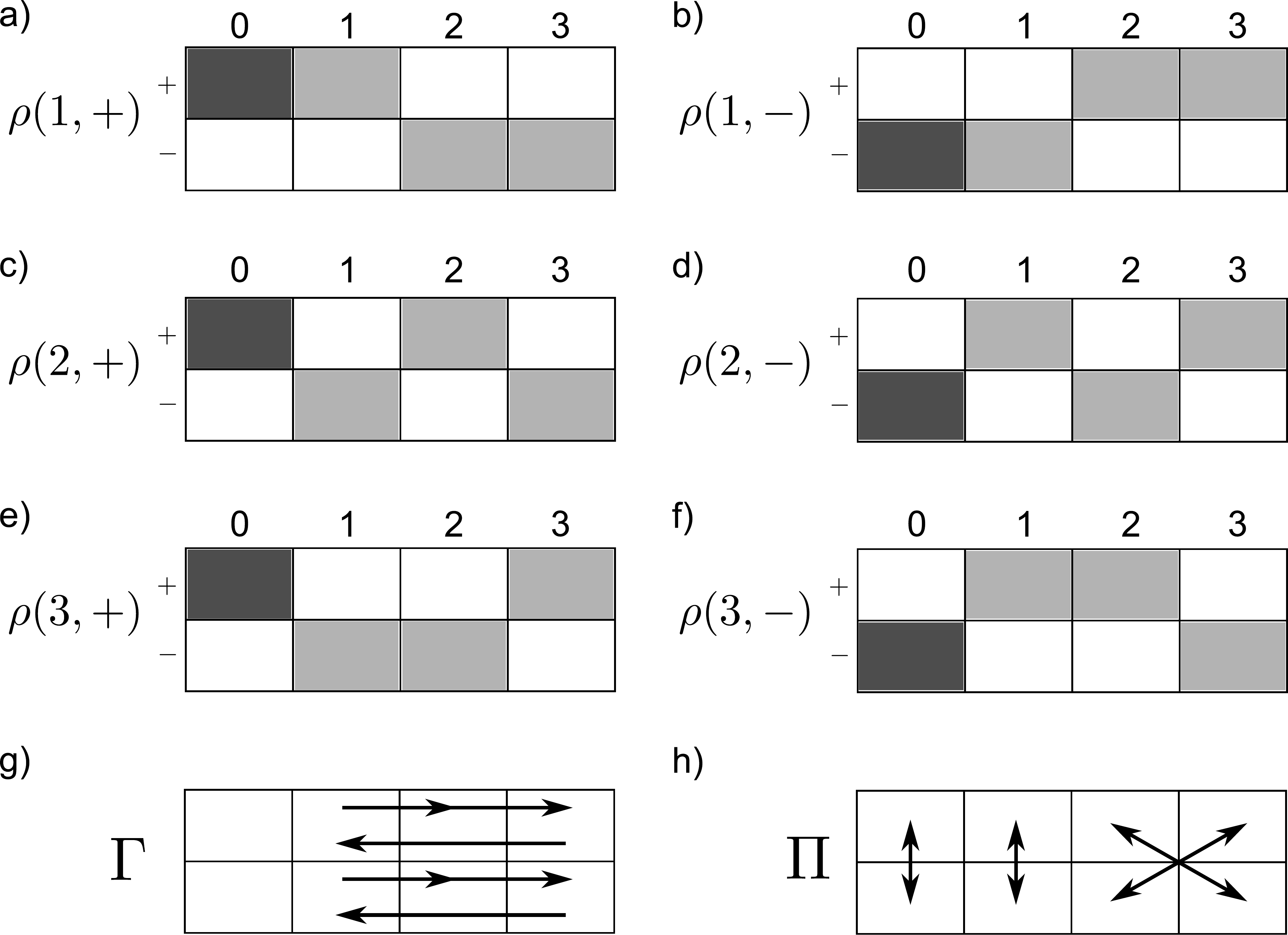}
	\caption{(a)-(f) Probability distributions $\mu_{\rho(j,\gamma)}(\epsilon,a)
$ over eight ontic states $\{(\epsilon,a):\epsilon=\pm,a\in\mbb{Z}_4\}$, for 
the three bases $\{\rho(j,\gamma):j=1,2,3,\gamma=\pm\}$ with Bloch vectors as 
in Eq.~\eqref{eq:D3_bases}, where the quantum states assign nonzero 
probability to the shaded ontic states. The shadings indicate the two 
nonzero probabilities $q_0$ (dark grey) and $q_1$ (light grey) that are assigned 
to the different ontic states. (g) and (h) Permutations of the ontic states 
that effect the rotations $\Gamma$ and $\Pi$, respectively.}
\label{fig:D3_transformations}
\end{figure}

To construct a quasiprobability distribution for which the three bases are 
non-negative and described by the above probability distributions, we need to 
find vectors $\vec{d}(a)$ satisfying Eq.~\eqref{eq:ontic_to_quasi}. By 
examining the permutation of ontic states that effects the transformation 
$\Gamma$, we note that $\vec{d}(a) = \Gamma_R^{a-1} \vec{d}(1)$ for $a=1,2,3$ 
[where $\Gamma_R$ is the fundamental (spin-$1$) representation of $\Gamma$], 
so we need only find $\vec{d}(0)$ and $\vec{d}(1)$, which can easily be 
determined to be
\begin{align}\label{eq:operator_vectors}
\vec{d}(0) &= \left(0,0,\sec\theta\right) \,,\nonumber\\
\vec{d}(1) &= \left(\tfrac{4}{3}\csc\theta,0,-\tfrac{1}{3}\sec\theta\right) \,.
\end{align}
The operators $F(\epsilon,a)$ are then
\begin{align}\label{eq:D3_distribution}
F(\epsilon,0) &= \frac{q_0}{2}\left(\unit + \gamma\sec\theta Z\right)	\,,
\nonumber\\
F(\epsilon,a) &= \frac{q_1}{2}\Gamma_U^{a-1}\left(\unit + \frac{4\gamma}{3}
\csc\theta X - \frac{\gamma}{3}\sec\theta Z\right)\Gamma_U^{1-a} \,,
\end{align}
for $\gamma=\pm$ and $a=1,2,3$, where
\begin{align}
\Gamma_U = \left(\begin{array}{cc} e^{\frac{2\pi i}{3}} & 0 \\ 0 & e^{-\frac{2
\pi i}{3}} \end{array}\right)
\end{align}
is the spin-$\frac{1}{2}$ representation of $\Gamma$. The operators in 
Eq.~\eqref{eq:D3_distribution} define a quasiprobability representation that 
reproduces all of quantum mechanics for the qubit and is non-negative for the 
states and measurements corresponding to the three bases with Bloch vectors 
in Eq.~\eqref{eq:D3_bases}.

\medskip\paragraph*{\textbf{Case 2.}} Up to an overall unitary, the second 
family of three bases have Bloch vectors
\begin{align}\label{eq:C2_bases}
\vec{r}(1) &= \left(\sin\theta,0,\cos\theta\right)	\,,\nonumber\\
\vec{r}(2) &= \left(-\sin\theta,0,\cos\theta\right)	\,,\nonumber\\
\vec{r}(3) &= \left(\cos\phi,\sin\phi,0\right)	\,,
\end{align}
for $\theta\in(0,\frac{\pi}{2})$ and $\phi\in(0,\pi)$, as illustrated in 
Fig.~\ref{fig:C2}. For $\phi\neq\frac{\pi}{2}$, the point group of these bases is 
$\mbb{Z}_2$ (i.e., the only nontrivial transformation is a $\pi$ rotation 
about the $z$-axis). If $\phi=\frac{\pi}{2}$, then the point group is the set 
of Pauli matrices, which is a projective representation of D$_2$ (or the 
octahedral group if $\theta=\frac{\pi}{4}$).

\begin{figure}[t!]
\centering
	\includegraphics[width=0.75\linewidth]{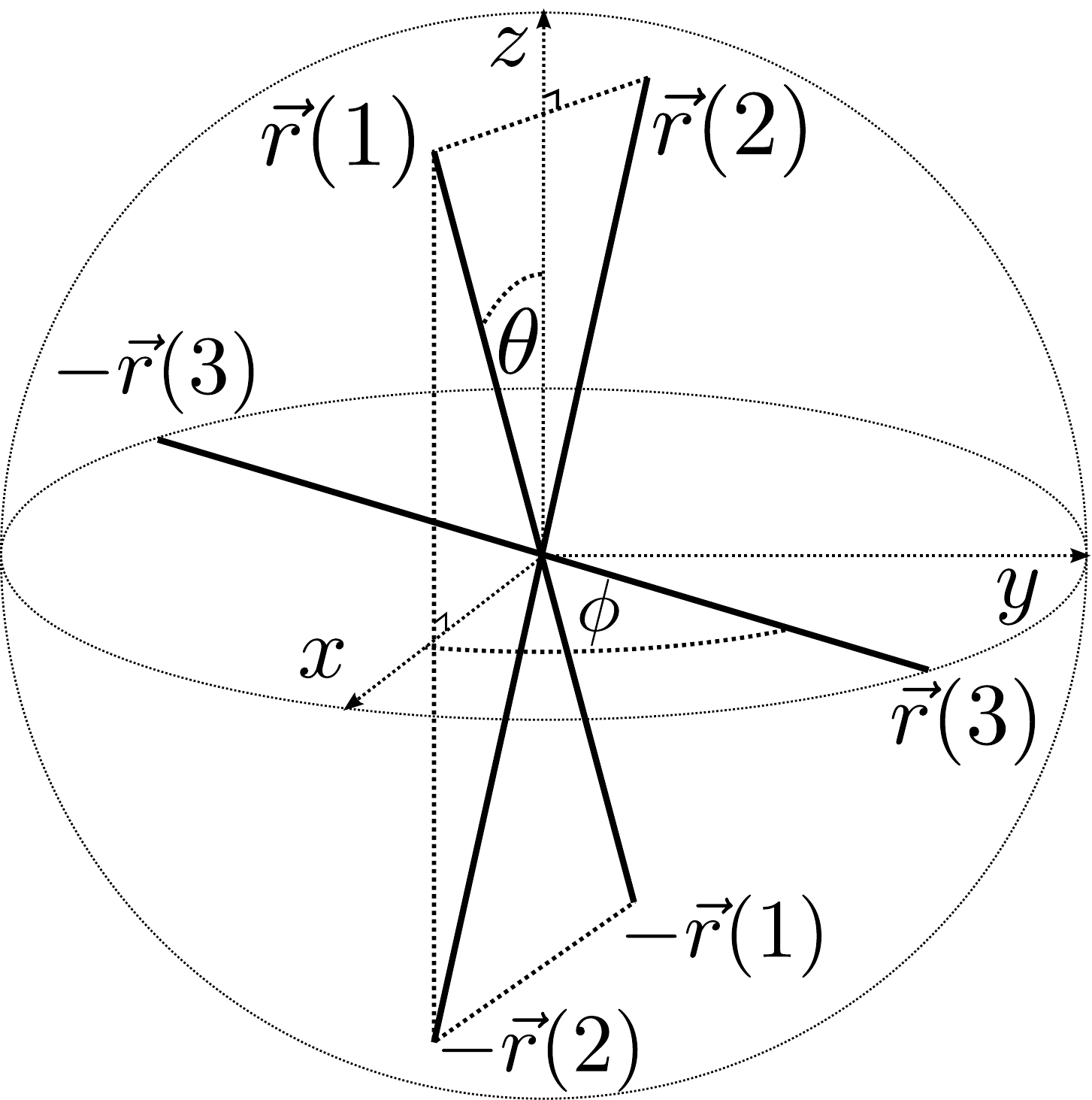}
	\caption{Illustration of the Bloch vectors in Eq.~\eqref{eq:C2_bases}, 
whose point group is $\mbb{Z}_2$.}\label{fig:C2}
\end{figure}

The only distributions $q(\epsilon,a)$ satisfying 
Eq.~\eqref{eq:normalizations} and \eqref{eq:3_bases_overlaps} are
\begin{align}
	q(+, 0) &= q(+, 3) := q_0	\,,\nonumber\\
	q(+, 1) &= q_0 - \frac{1}{2}\cos2\theta - \frac{1}{2}\cos\phi\sin\theta	\,,
\nonumber\\
	q(+, 2) &= q_0 - \frac{1}{2}\cos2\theta + \frac{1}{2}\cos\phi\sin\theta	\,,
\nonumber\\
	q(-, 0) &= q(-, 3) = \frac{1}{2} - q_0 + \frac{1}{2}\cos2\theta	\,,
\nonumber\\
	q(-, 1) &= \frac{1}{2} - q_0 - \frac{1}{2}\cos\phi\sin\theta	\,,\nonumber\\
	q(-, 2) &= \frac{1}{2} - q_0 + \frac{1}{2}\cos\phi\sin\theta	\,,
\end{align}
where $q_0\in[0,1]$ is a free parameter. Requiring all probabilities to be in 
the interval $[0,1]$, we obtain
\begin{align}
0 \leq q(+,1) + q(-,1) &= \sin\theta(\sin\theta-\cos\phi) \leq 2	\,, 
\nonumber\\
0 \leq q(+,2) + q(-,2) &= \sin\theta(\sin\theta+\cos\phi) \leq 2	\,.
\end{align}
Therefore we require $|\cos\phi| \leq \sin\theta$ in order for the bases in 
Eq.~\eqref{eq:C2_bases} to be non-negative in a quasiprobability 
representation. This is another example of a more restrictive condition on
non-negative bases than either Theorem~\ref{thm:qubit_planar} or 
\ref{thm:qubit_new}. This condition is also sufficient, as for such $\theta,\phi$, we can 
set the distribution to
\begin{align}\label{eq:C2_distribution}
q_0 := q(\gamma,0) &= q(\gamma,3) = \frac{\cos^2\theta}{2}	\,,\nonumber\\
q_1 := q(\gamma,1) &= \frac{\sin\theta}{2}(\sin\theta - \cos\phi)	\,,
\nonumber\\
q_2 := q(\gamma,2) &= \frac{\sin\theta}{2}(\sin\theta + \cos\phi)	\,,
\end{align}
for $\gamma=\pm$. 

As with the models for the bases in Eq.~\eqref{eq:D3_bases}, we can view the 
unitary transformations of quantum states as supervening on permutations of 
ontic states. The distribution in Eq.~\eqref{eq:C2_distribution} captures the 
symmetry of the bases and the possible transformations, as unitary 
transformations can only supervene on permutations of the ontic states for 
values of $\theta,\phi$ for which the bases have the appropriate symmetry.

For $\phi\neq\frac{\pi}{2}$, the only nontrivial unitary transformation that 
permutes non-negative states is $Z$ (i.e., a $\pi$-flip about the $z$-axis). 
The permutation of ontic states that effects $Z$ is depicted in 
Fig.~\ref{fig:C2_transformations} (g). 

For $\phi=\frac{\pi}{2}$, the symmetry group is the group of Pauli matrices, 
generated by $X$ and $Z$, which is a spin-$\frac{1}{2}$ representation of D$_2
$. The permutation of ontic states that effects $X$ is depicted in 
Fig.~\ref{fig:C2_transformations} (h). This permutation maps $(+,1)$ to 
$(-,2)$, and so we require $q(+,1) = q(-,2)$, which is satisfied if and 
only if $\phi=\frac{\pi}{2}$ in the distribution in 
Eq.~\eqref{eq:C2_distribution}. 

\begin{figure}[t!]
\centering
	\includegraphics[width=0.95\linewidth]{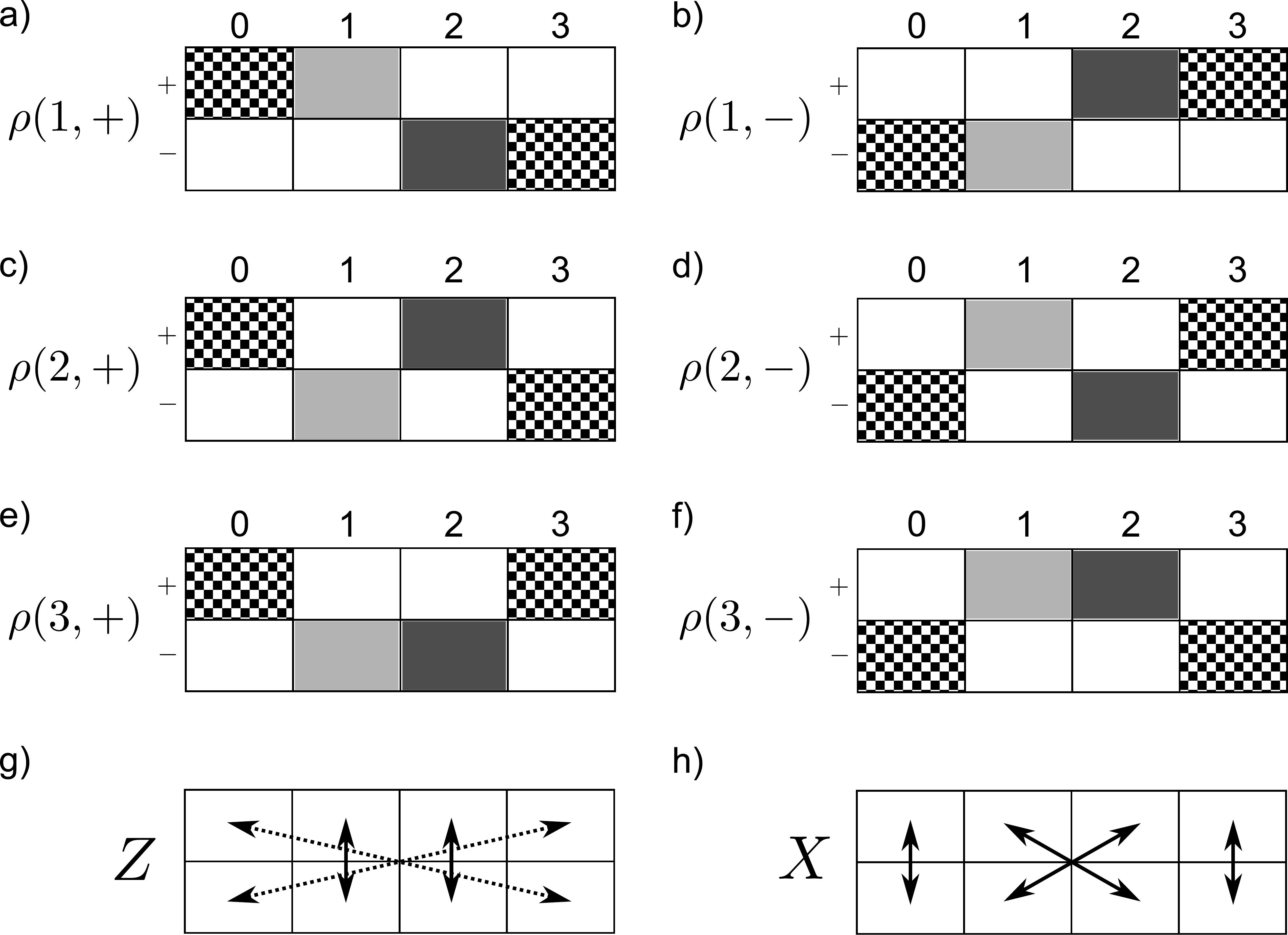}
	\caption{(a)-(f) Probability distributions $\mu_{\rho(j,\gamma)}(\epsilon,a)
$ over eight ontic states $\{(\epsilon,a):\epsilon=\pm,a\in\mbb{Z}_4\}$, for 
the three bases with Bloch vectors as in Eq.~\eqref{eq:C2_bases}. The 
shadings indicate the three nonzero probabilities, $q_0$ (checkered), $q_1$ (
light grey) and $q_2$ (dark grey), that are assigned to the different ontic 
states. (g) and (h) Permutations of the ontic states that effects the unitary 
transformations $Z$ and $X$, respectively. Note that the permutation in (h) 
is only valid for the special case where the bases have D$_2$ symmetry (i.e., 
when $q_1 = q_2$).}\label{fig:C2_transformations}
\end{figure}

To construct a quasiprobability distribution for which these bases are 
non-negative, we need to find vectors $\vec{d}(a)$ satisfying 
Eq.~\eqref{eq:ontic_to_quasi}. One such set of vectors is
\begin{align}
\vec{d}(0) &= \left(0,\csc\phi,\sec\theta\right) \,,\nonumber\\
\vec{d}(1) &= \left(\csc\theta,{-}\csc\phi - \cos\phi\csc\theta,0\right) \,,\nonumber\\
\vec{d}(2) &= \left({-}\csc\theta,{-}\csc\phi + \cos\phi\csc\theta,0\right) \,,\nonumber\\
\vec{d}(3) &= \left(0,\csc\phi,{-}\sec\theta\right) \,.
\end{align}
The corresponding operators $F(\gamma,k)$ are then
\begin{align}\label{eq:C2_operators}
F(\gamma, 0) &= \frac{q_0}{2}\left(\unit + \gamma\csc\phi Y + \gamma\sec\theta
Z \right)	\,,\nonumber\\
F(\gamma, 1) &= \frac{q_1}{2}\left[\unit + \gamma\csc\theta X - \gamma(
\csc\phi + \cos\phi\csc\theta) Y \right]	\,,\nonumber\\
F(\gamma, 2) &= \frac{q_2}{2}\left[\unit - \gamma\csc\theta X - \gamma(
\csc\phi - \cos\phi\csc\theta) Y \right]	\,,\nonumber\\
F(\gamma, 3) &= \frac{q_0}{2}\left(\unit + \gamma\csc\phi Y - \gamma\sec\theta
Z \right)	\,,
\end{align}
for $\gamma = \pm$. These operators define a quasiprobability representation 
that reproduces all of quantum mechanics for the qubit and is non-negative 
for the states and measurements corresponding to the three bases with Bloch 
vectors in Eq.~\eqref{eq:C2_bases}.

\medskip\textit{\textbf{Stabilizer states.}} We now consider the common 
special case of both families of three bases. For this special case, the 
three bases are equivalent to the vertices of a regular octahedron (i.e., 
single qubit stabilizer states) and so have a larger symmetry group, namely, 
the octahedral group (equivalent to the single qubit Clifford group). In the 
appropriate special cases, the distributions $q(\epsilon,a)$ in 
Eq.~\eqref{eq:D3d_probabilities} ($\sin^2\theta = \tfrac{2}{3}$) and 
Eq.~\eqref{eq:C2_distribution} ($\phi=\frac{\pi}{2}$ and $\theta=\frac{\pi}{4}$)
become uniform and $q(\epsilon,a)=\tfrac{1}{4}$ for all $\epsilon,a$.

The single qubit Clifford group is generated by the Hadamard gate $H$ and the 
phase gate $P$. These transformations supervene on the permutations of ontic 
states shown in Fig.~\ref{fig:clifford_transformations} in the basis 
$\vec{r}(1) = \vec{x}$, $\vec{r}(2) = \vec{y}$ and $\vec{r}(3) = \vec{z}$. Thus, 
all Clifford transformations can be viewed as supervening on permutations of 
ontic states in these models.

\begin{figure}[t!]
\centering
	\includegraphics[width=0.8\linewidth]{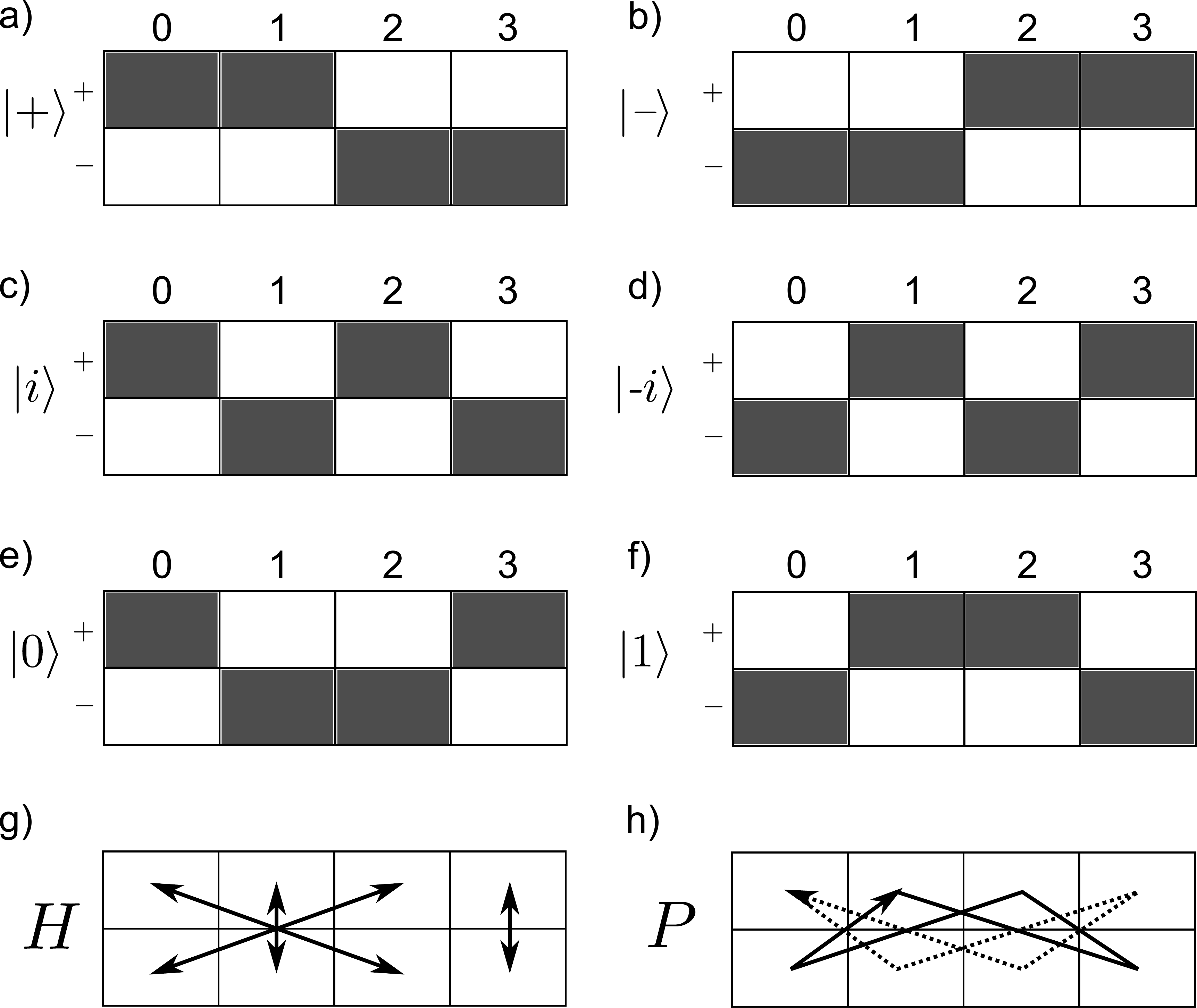}
	\caption{(a)-(f) Probability distributions $\mu_{\rho}(\epsilon,a)$ over 
eight ontic states $\{(\epsilon,a):\epsilon=\pm,a\in\mbb{Z}_4\}$ for the 
single qubit stabilizer states. The shading indicates the ontic states that 
are assigned probability $\tfrac{1}{4}$ for a given stabilizer state. (g) and 
(h) Permutations of the ontic states that effect the unitary transformations 
$H$ and $P$, respectively.}\label{fig:clifford_transformations}
\end{figure}

\subsubsection*{General features of quasiprobability representations with 
three non-negative bases}

We now briefly discuss some of the features of these quasiprobability 
representations, and how they relate to Spekkens' toy 
theory~\cite{Spekkens2007} and the discrete Wigner function~\cite{Gibbons2004}.

Studying non-negative bases that are closed under a group of nontrivial 
unitary transformations has shown that Theorems~\ref{thm:qubit_planar} and 
\ref{thm:qubit_new} do not completely characterize the possible sets of 
single qubit bases that can be non-negative in some quasiprobability 
distribution, as neither theorem excludes the bases in Eq.~\eqref{eq:D3_bases}
for $\sin^2\theta>\frac{8}{9}$ or in Eq.~\eqref{eq:C2_bases} for $1>|\cos\phi
|>\sin\theta$, for which the bases are ``close'' to coplanar. This suggests 
that models for sets of three bases only exist when they are sufficiently far 
from being coplanar. However, as an exception to this, the bases in 
Eq.~\eqref{eq:D3_bases} still admit a noncontextual model when $\sin\theta\to0$, 
i.e, when all three bases are close to being degenerate. 

Note also that whenever a quasiprobability representation can be defined such 
that the three bases are non-negative, then any unitary transformations that 
permute non-negative states can be interpreted as supervening on permutations 
of the ontic states.

For the remainder of this discussion, we focus on our quasiprobability 
representation of the stabilizer states. Our quasiprobability representation 
of stabilizer states is related to the standard definition of the discrete 
Wigner function, though our representation is defined over eight points 
rather than four, as follows. Note that in our quasiprobability 
representation, each non-negative state is uniquely defined by its 
distribution over the reduced phase space $\{(+,a):a\in\mbb{Z}_4\}$. Over 
this reduced phase space, the operators in either 
Eq.~\eqref{eq:D3_distribution} or Eq.~\eqref{eq:C2_operators} 
can be written as
\begin{align}
F(+,a) = c\sum_j \rho(j,\gamma_{j,a}) - d\unit	\,,
\end{align}
for $a\in\mbb{Z}_4$, where $\gamma_{j,a}$ is chosen such that 
$(+,a)\in\supps{\rho(j,\gamma_{j,a})}$. By ignoring the points 
outside the reduced phase space and doing suitable renormalizations 
(i.e., setting $c=d=\frac{1}{2}$), we recover a typical definition 
of a discrete Wigner function~\cite{Gibbons2004}. Note that in this 
setting, the Clifford transformations permute the non-negative basis, 
but do not supervene on permutations of ontic states.

Our non-negative quasiprobability representation of preparations and 
measurements in the stabilizer bases, together with Clifford transformations 
amongst these states, also recovers those of Spekkens' toy 
theory~\cite{Spekkens2007} when restricted to this reduced phase space. 
In Spekkens' theory, one cannot define suitable ontic transformations on 
which all single qubit Clifford transformations supervene. For our theory
restricted to the reduced phase space, this is also the case. For example, 
the permutation of ontic states that effects $\Pi$ [depicted in 
Fig.~\ref{fig:D3_transformations} (h)] mixes the values of $\epsilon$ and 
so cannot be defined on this reduced ontological space. Furthermore, 
anti-unitary transformations can supervene on permutations of ontic states, 
as is the case in the toy theory.

Unlike the toy theory, our model of the stabilizer states over the full phase 
space of 8 points allows all Clifford transformations to supervene on 
permutations of ontic states. One can ask the additional question of whether 
all permutations of ontic states that permute the non-negative states are 
allowed in the theory. If one requires that only unitary transformations 
supervene on ontic permutations, then the answer is no. However, if one 
allows antiunitary transformations, provided they permute the non-negative 
states, then all unitary and antiunitary transformations that permute the 
stabilizer states supervene on permutations of the ontic states. In 
particular, the permutation that maps $\epsilon\to-\epsilon$ and leaves $a$ 
unchanged effects a universal NOT gate.

\subsubsection*{Four non-negative bases}

The Bloch vectors 
\begin{align}\label{eq:four_bases}
\vec{r}(\pm,\pm,\pm) := \left(\pm\cos\phi\sin\theta,\pm\sin\phi\sin\theta,
\pm\cos\theta\right)
\end{align}
for $\theta,\phi\in [0,\tfrac{\pi}{2}]$, corresponding to the vertices of a 
right cuboid give a set of four bases. Theorem~\ref{thm:qubit_new} shows that 
any set of four non-negative bases in an arbitrary quasiprobability 
representation must be of this form (up to unitary equivalence). 

The high degree of symmetry makes it easy to construct a quasiprobability 
representation that is non-negative for the above bases for any values of $
\theta,\phi\in(0,\tfrac{\pi}{2})$. By examining Eq.~\eqref{eq:ontic_to_quasi}, 
it is apparent that up to an overall permutation of $\Lambda$, we have
\begin{align}
\vec{d}_a = \frac{\vec{e}_a}{\vec{e}_a\cdot\vec{r}_{+++}}
\end{align}
where the $\vec{e}_a$ are the canonical basis vectors of $\mathbb{R}^3$.

We therefore consider a quasiprobability representation over six points, 
$\{(\epsilon,a):\epsilon=\pm,a=1,2,3\}$. Solving Eqs.~\eqref{eq:normalizations} 
and \eqref{eq:3_bases_overlaps} for four non-negative bases gives
\begin{align}
q(\epsilon,1) &=: q_1 = 1 + \sin^2\theta\cos2\phi - \cos^2\theta	\,,
\nonumber\\
q(\epsilon,2) &=: q_2 = 1 - \sin^2\theta\cos2\phi - \cos^2\theta	\,,
\nonumber\\
q(\epsilon,3) &=: q_3 = 1 + \cos2\theta	\,.
\end{align}
for $\epsilon=\pm$. Unlike the quasiprobability representation for three 
non-negative bases, this quasiprobability representation is valid for all $\theta,
\phi\in(0,\tfrac{\pi}{2})$, even when the four non-negative bases are 
arbitrarily close to being coplanar.

The point group of the four bases for general $\theta$ and $\phi$ is the 
group of Pauli matrices. The supports of the non-negative states and the 
permutations of ontic states corresponding to $X$ and $Z$ are illustrated in 
Fig.~\ref{fig:support_cube} (a)-(j). 

There are three equivalent cases with slightly higher symmetry, corresponding 
to when two parallel faces of the cuboid are squares. For example, the face in 
the $z$ plane is a square when $\phi=\tfrac{\pi}{4}$. In this case, the bases 
are also invariant under the phase gate, $P$, as shown in 
Fig.~\ref{fig:support_cube} (l). 

In the special case of the cube (which occurs when $\phi=\tfrac{\pi}{4}$, $
\cos\theta=\frac{1}{\sqrt{3}}$), the symmetry group of the four bases is the 
octahedral group. The permutations of ontic states corresponding to two 
generators of the single qubit Clifford group (the Hadamard $H$ and $P$) are 
illustrated in Fig.~\ref{fig:support_cube} (k) and (l) respectively. 

Note that the bases in Eq.~\eqref{eq:D3_bases} with $\sin^2\theta=\frac{8}{9}$
(i.e., $\cos\theta=\pm\frac{1}{3}$) are three of the four bases 
corresponding to the vertices of the cube, so in this limiting case, a fourth 
non-negative basis (with Bloch vectors parallel to the $z$-axis) can be 
added. For the model in Eq.~\eqref{eq:D3d_probabilities} in the limiting case 
$\sin^2\theta=\frac{8}{9}$, $q_0=0$ and so only six points are assigned 
nonzero probability. If the points $(\pm,0)$, to which any non-negative state 
assigns zero probability, are ignored, then the distributions in 
Fig.~\ref{fig:D3_transformations} (a)-(f) are identical 
(up to a relabeling of the states) to those in 
Fig.~\ref{fig:support_cube} (c)-(h). 

\begin{figure}
\centering
	\includegraphics[width=0.95\linewidth]{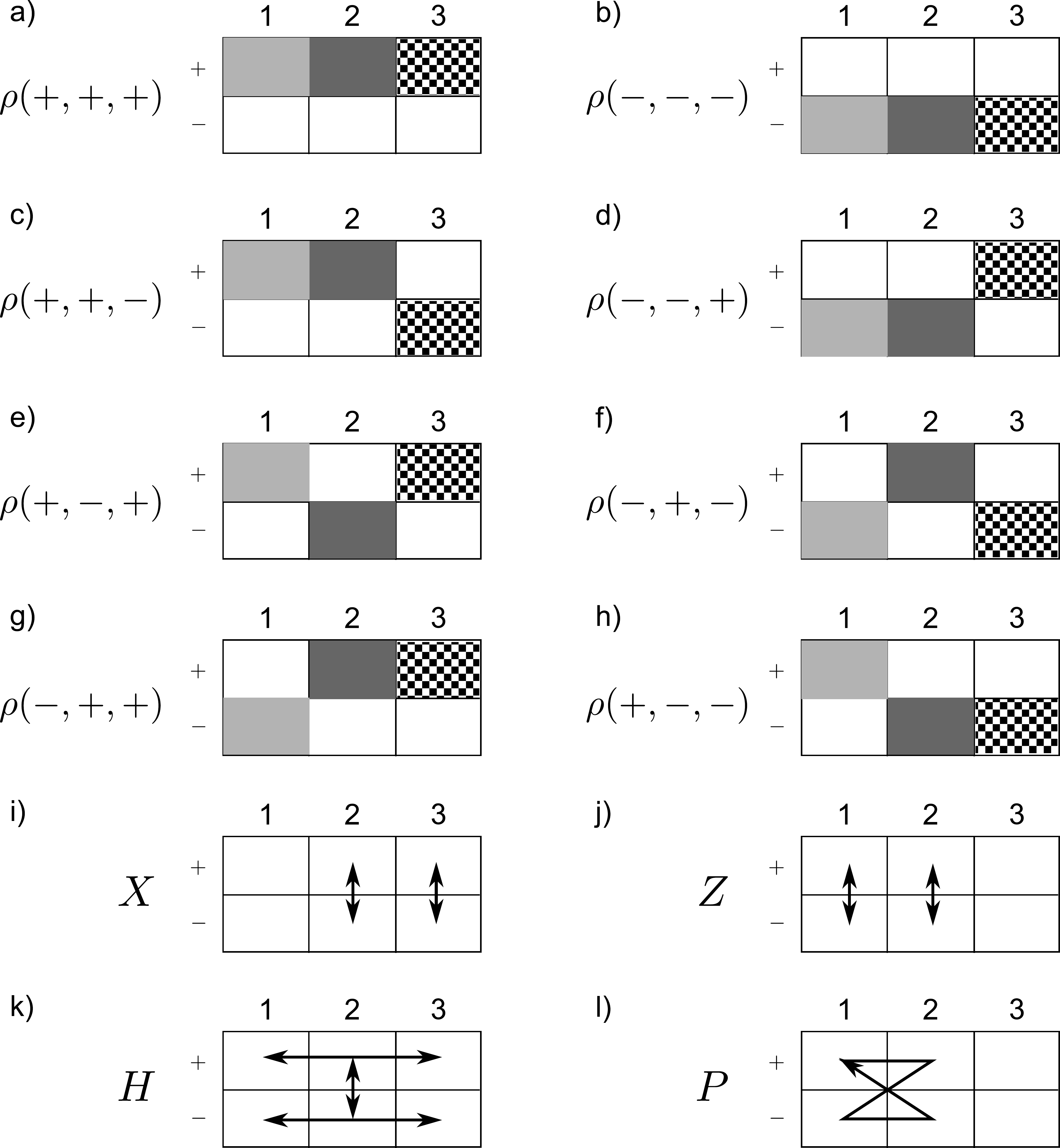}
	\caption{(a)-(h) Probability distributions over 6 points $\{(\gamma,k):
\gamma=\pm,k=1,\ldots,3\}$, for the eight quantum states with Bloch vectors 
as in Eq.~\eqref{eq:four_bases} corresponding to the vertices of a right 
cuboid. The shadings indicate the three nonzero probabilities, $q_1$ (light 
grey), $q_2$ (dark grey) and $q_3$ (checkered), that are assigned to the 
different ontic states. (i)-(l) Permutation of ontic states that effect the 
transformations (i) $X$, (j) $Z$, (k) $H$, and (l) $P$ respectively. Note 
that the permutations in (k) and (l) are only valid for special cases with 
higher symmetry, namely, when $q_1=q_3$ and $q_1=q_2$ respectively.}
\label{fig:support_cube}
\end{figure}

\subsubsection*{Example of bases that do not admit a preparation 
noncontextual model: Icosahedron}

We have provided an exhaustive list of sets of bases with nontrivial point 
groups that are non-negative in some quasiprobability distribution and 
noticed that the single qubit Clifford group (i.e., the symmetry group of a 
regular octahedron) plays a special role.

The octahedral group is one of two groups that permute pairs of antipodal 
points that are the vertices of a Platonic solid, the other being the 
icosahedral group. While Theorems~\ref{thm:qubit_planar} and 
\ref{thm:qubit_new} show that the pairs of antipodal vertices 
corresponding to either the icosahedron (or its dual, the dodecahedron) 
cannot all correspond to non-negative bases, we provide an explicit 
proof of this fact, which is illustrative of the general case.

The idea behind the proof is to use the Bloch vectors corresponding to the 
vertices of an icosahedron to obtain two decompositions of a mixed state $\rho
$ in terms of two different sets of states and then use 
Lemmas~\ref{lem:sum_prob}--\ref{lem:common_support} to show that 
in any quasiprobability representation, some of the quasi-probabilities 
must be negative. 

The twelve vertices of an icosahedron are given by the six vectors
\begin{align}\label{eq:icosahedron}
\vec{r}(0,\alpha) &= \frac{1}{\sqrt{1+\varphi^2}}\left(1, \alpha\varphi, 0
\right) \,,\nonumber\\
\vec{r}(1,\alpha) &= \frac{1}{\sqrt{1+\varphi^2}}\left(0, 1, 
\alpha\varphi\right) \,,\nonumber\\
\vec{r}(2,\alpha) &= \frac{1}{\sqrt{1+\varphi^2}}\left(\alpha\varphi, 0, 1
\right) \,,
\end{align}
for $\alpha = \pm$ and their negatives, where $\varphi = (1+\sqrt{5})/2$ is 
the golden ratio. These vertices satisfy the relation
\begin{align}
(1-2a)\vec{r}(1,+) = a\vec{r}(0,+) - a\vec{r}(0,-) - (1-2a)\vec{r}(1,-)	\,,
\end{align}
where $a = (2 + \varphi)^{-1}$. Therefore we have
\begin{align}
\rho &=	a\rho(0,+,+) + a\rho(0,-,-) + (1-2a)\rho(1,-,-) \nonumber\\
	&= (1-a)\rho(1,+,+) + a\rho(1,+,-)	\,,
\end{align}
where $\rho(j,\alpha,\mu) = \frac{1}{2}(\unit + \mu\vec{r}_{j,\alpha})$. For 
a quasiprobability distribution, this implies
\begin{align}
\mu_{\rho}(\lambda) &=	a\mu_{\rho(0,+,+)}(\lambda) + a\mu_{\rho(0,-,-)}(
\lambda) + (1-2a)\mu_{\rho(1,-,-)}(\lambda) \nonumber\\
	&= (1-a)\mu_{\rho(1,+,+)}(\lambda) + a\mu_{\rho(1,+,-)}(\lambda)	\,,
\end{align}
for all $\lambda\in\Lambda$. Assume that all the bases corresponding to the 
vertices of an icosahedron are non-negative. By Lemma~\ref{lem:common_support},
there exists a $\lambda'\in\supps{\rho(1,+,-)}\cap\supps{\rho(1,-,-)}$. By 
Lemma~\ref{lem:sum_prob}, for this value of $\lambda'$, the first line is at 
least $(1-2a)q(\lambda')>0$, while the second is exactly $aq(\lambda')>0$. As 
$(1-2a)>a$, this is a contradiction. Therefore not all the bases can be 
non-negative.

\section{Quasi-probability representations of qudits}
\label{sec:higher_dimensions}

For qubits, we have shown that no quasiprobability distribution can be 
non-negative for more than four orthonormal bases. Moreover, up to unitary 
equivalence, the only set of four bases that can be non-negative corresponds 
to the vertices of a right cuboid. We now turn our attention to the general 
case and establish restrictions on the relation between non-negative bases in 
an arbitrary quasiprobability representation of $\mc{H}_d$. We will begin by 
constructing a mixed state $\rho$ that plays an analogous role to the states 
in Eq.~\eqref{eq:coplanar_bases_decomposition} and \eqref{eq:decomposition} 
for qubits. We then generalize Theorems~\ref{thm:qubit_planar} and 
\ref{thm:qubit_new} to higher dimensions, which establishes restrictions 
on the relation between any set of non-negative bases. We conclude
this section by obtaining an upper bound of $2^{d^2}$ on the number
of states that are elements of non-negative bases in an arbitrary 
quasiprobability representation.

For higher dimensions, the relationship between states and orthonormal bases 
is not as straightforward as for qubits, where a state can be uniquely 
extended to an orthonormal basis. Moreover, for arbitrary $d$, there is no 
simple map from $\mc{H}_d$ to an intuitive geometric space like the Bloch 
sphere. Therefore we rely on abstract algebraic tools rather than geometric 
ones. We note that this section is quite technical relative to previous 
sections of this paper, and can be skipped upon first reading.

\subsection{Linear algebra and convex geometry}

We now introduce some tools from linear algebra. These tools enable us to 
deduce the existence of mixed states with multiple decompositions in terms of 
a set of non-negative bases if the bases are related in a particular way. 
Such states play an analogous role to the states in 
Eq.~\eqref{eq:coplanar_bases_decomposition} and \eqref{eq:decomposition} for qubits.

Two density matrices, $\rho,\omega\in\mc{B}(\mc{H}_d)$, are \textit{linearly 
independent} (over $\mathbb{R}$) if
\begin{align}
a\rho + b\omega = 0 \Rightarrow a = b = 0	\,,
\end{align}
and are \textit{linearly dependent} otherwise. Note that we are considering 
linear independence over the operator space $\mathcal{B}(\mathcal{H}_d)$, 
rather than the space of pure states. We introduce the term \textit{disparate}
to denote a set $\mc{A}=\{\rho(\alpha,j):\alpha\in\mbb{Z}_N,j\in\mbb{Z}_d\}$ 
of $N$ bases of $\mc{H}_d$ such that any set $\mc{A}'$ obtained from $\mc{A}$ 
by removing one element from each basis and including the identity is a set 
of linearly independent operators. That is, for all $\vec{f}\in\mbb{Z}_d^N$ (
where $f_j$ labels the element of the $j$th basis that is not in $\mc{A}'$),
\begin{align}
\sum_{\alpha,j|j\neq f_{\alpha}} p_{\alpha,j}\rho(\alpha,j) = \frac{b}{d}
\unit	\Rightarrow p_{\alpha,j} = b = 0 \ \forall (\alpha,j)	\,.
\end{align}
The maximum number of disparate bases of $\mc{H}_d$ is $d+1$, as each basis 
contributes $d-1$ linearly independent density operators and including the 
identity (or the remaining element of any of the bases) gives $d^2$ linearly 
independent density operators. 

We denote the \textit{span} of a set of density matrices $\mc{A}=\{\rho(j)\}$ 
over $\mathbb{R}$ by
\begin{align}
\mathcal{T}(\mc{A}) = \{\sum_j c_j \rho(j):c_j\in\mathbb{R}\}	\,.
\end{align}
The convex hull of a set of density matrices $\mc{A}$ is the set
\begin{align}
\mathcal{C}(\mc{A}) = \{\sum_j p_j \rho(j):p_j\geq 0,\sum_j p_j = 1\}	\,.
\end{align}
Denote by $\partial\mathcal{C}(\mc{A})$ the surface of $\mathcal{C}(\mc{A})$ 
with respect to $\mathcal{T}(\mc{A})$. Then, as we now show, if $\mc{A}$ is a 
set of disparate bases, any expansion of any $\rho\in\partial\mathcal{C}(\mc{A})$ 
as a convex combination of the elements of $\mc{A}$ must assign a zero 
coefficient to at least one element of each of the bases in $\mc{A}$.

\begin{lem}\label{lem:convex_hull}
Let $\mc{A}=\{\rho(\alpha,j):j\in\mbb{Z}_d,\alpha\in\mbb{Z}_N\}$ be a set of 
$N$ disparate bases. Then for all $\rho\in\partial\mathcal{C}(\mc{A})$, there 
exist $\vec{f}\in\mbb{Z}_d^N$ and $p_{\alpha,j}\geq 0$ such that
\begin{align}
\rho = \sum_{\alpha,j|j\neq f_{\alpha}} p_{\alpha,j}\rho(\alpha,j)	\,.
\end{align}
\end{lem}

\begin{proof}
To prove this lemma, we prove the contrapositive. Let $\rho\in\mathcal{C}(\mc{
A})$ be such that all the coefficients $\{p_{\alpha,j}:j\in\mathbb{Z}_d\}$ in 
some decomposition
\begin{align}\label{eq:rho_interior}
\rho = \sum_{\alpha,j} p_{\alpha,j}\rho(\alpha,j)
\end{align}
are nonzero for some value of $\alpha$. For each $\alpha\in\mbb{Z}_N$, we 
can define $f_{\alpha}$ such that
\begin{align}
p_{\alpha,f_{\alpha}} = \min_{j\in\mbb{Z}_d} p_{\alpha,j}	\,.
\end{align}
Then we can rewrite $\rho$ as
\begin{align}
\rho = \sum_{\alpha,j|j\neq f_{\alpha}} (p_{\alpha,j}-p_{\alpha,f_{\alpha}})
\rho(\alpha,j) + \unit\sum_{\alpha} p_{\alpha,f_{\alpha}}	\,.
\end{align}
Therefore we can set all the $p_{\alpha,f_{\alpha}}$ to be equal and strictly 
positive without loss of generality.

We now prove that $\rho$ is in the interior of $\mathcal{C}(\mc{A})$ with 
respect to $\mathcal{T}(\mc{A})$. Let $\rho'\in\mathcal{T}(\mc{A})$. Then, by 
definition, there exists $\beta_{\alpha,j}\in\mathbb{R}$ with $\sum_{\alpha,j}
\beta_{\alpha,j} = 1$ such that
\begin{align}\label{eq:decomposition_rho}
\rho' = \sum_{\alpha,j}\beta_{\alpha,j}\rho(\alpha,j)	\,.
\end{align}
If we set 
\begin{align}
\delta = \min_{\rho'\in\mathcal{T}(\mc{A})}\min_{\alpha,j}\beta_{\alpha,j} <0 
\,,
\end{align}
then for all $\rho'\in\mathcal{T}(\mc{A})$, the coefficients $\epsilon\beta_{
\alpha,j} + (1-\epsilon)p_{ij}$ in the decomposition
\begin{align}
\epsilon\rho' + (1-\epsilon)\rho = \sum_{\alpha,j}(\epsilon\beta_{\alpha,j} + 
(1-\epsilon)p_{\alpha,j})\rho(\alpha,j)	\,,
\end{align}
are non-negative whenever
\begin{align}
\beta_{\alpha,j}\geq0	\quad \text{or} \quad \epsilon \leq \frac{p_{i,j}}{p_{
i,j} - \beta_{\alpha,j}} 	\,.
\end{align}
Therefore if we set
\begin{align}
\epsilon\leq\min_{i,j}\frac{p_{i,j}}{p_{i,j} - \delta}	\,,
\end{align}
then $\epsilon\rho + (1-\epsilon)\rho'\in\mathcal{C}(\mc{A})$ for all $\rho'
\in\mathcal{T}(\mc{A})$.
\end{proof}

Lemma~\ref{lem:convex_hull} shows that any point $\rho$ on the surface of the 
convex hull of a set of disparate bases can be written as
\begin{align}
\rho = \sum_{\alpha,j|j\neq f_{\alpha}} p_{\alpha,j}\rho(\alpha,j)
\end{align}
for some $\vec{f}\in\mbb{Z}_d^N$ and $p_{\alpha,j}\geq 0$. Therefore for all 
pure states $\phi\in\mathcal{T}(\mc{A})$ [which must be on the boundary of 
$\mathcal{T}(\mathcal{A})$], the lines defined by
\begin{align}
\epsilon\phi + \frac{1-\epsilon}{d}\unit
\end{align}
for $\epsilon\in[0,1]$ must intersect $\partial\mathcal{C}(\mc{A})$ as 
$\frac{1}{d}\unit$ is in the interior of $\mathcal{C}(\mc{A})$ and $\phi$ is 
either in $\partial\mathcal{C}(\mc{A})$ or not in $\mathcal{C}(\mc{A})$. 
Therefore for all pure states $\phi\in\mathcal{T}(\mc{A})$, there exists 
$\epsilon>0$, $\vec{f}\in\mbb{Z}_d^N$ and $p_{\alpha,j}\geq 0$ such that
\begin{align}\label{eq:decompose_state}
\rho:= \epsilon\phi + \frac{1-\epsilon}{d}\unit = \sum_{\alpha,j|j\neq f_{\alpha}} p_{\alpha,j}\rho(\alpha,j)	\,.
\end{align}

The state defined in Eq.~\eqref{eq:decompose_state} can be decomposed as a 
convex combination of different bases in multiple ways, which allows us to 
use it to restrict the relation between non-negative bases in an analogous 
way to the states in Eq.~\eqref{eq:coplanar_bases_decomposition} and 
\eqref{eq:decomposition} for qubits. 

\subsection{Relation between non-negative bases for qudits}\label{sec:qudits}

We now use the state $\rho$ defined in Eq.~\eqref{eq:decompose_state} to 
generalize Theorem~\ref{thm:qubit_planar} to qudits.

\begin{thm}\label{thm:qudits_coplanar}
In an arbitrary quasiprobability representation, any set of three mutually 
non-orthogonal non-negative bases must be disparate.
\end{thm}

\begin{proof}
Let $\{\rho(\alpha,j):\alpha\in\mbb{Z}_3,j\in\mbb{Z}_d\}$ be a set of three 
mutually non-orthogonal bases that are not disparate and are non-negative in 
some quasiprobability representation. 

As the three bases are not disparate, we can relabel the elements of 
$\{\rho(2,j):j\in\mbb{Z}_d\}$ such that $\rho(2,0)\in\mathcal{T}(\mc{A})$, where 
$\mc{A} = \{\rho(\alpha,j):\alpha\in\mbb{Z}_2,j\in\mbb{Z}_d\}$ is the set of the 
elements of the other two bases.

Therefore we can use the decomposition in Eq.~\eqref{eq:decompose_state} and 
relabel the elements of the first two bases such that $\vec{f} = (0,0)$ to 
obtain
\begin{align}
\epsilon\rho(2,0) + \frac{1-\epsilon}{d}\unit = \sum_{\alpha=0}^1\sum_{j=1}^d 
p_{0,j}\rho(\alpha,j)
\end{align}
for some $\epsilon\in (0,1)$. As $\mu$ is convex-linear, we have
\begin{align}\label{eq:3_nondisparate}
\epsilon\mu_{\rho(2,0)}(\lambda) + \frac{1-\epsilon}{d}q(\lambda) = \sum_{
\alpha=0}^1\sum_{j=1}^d p_{\alpha,j}\mu_{\rho(\alpha,j)}(\lambda)
\end{align}
for all $\lambda\in\Lambda$. As $\rho(0,0)$ and $\rho(1,0)$ are not 
orthogonal, there exists $\lambda'\in\supps{\rho(0,0)}\cap\supps{\rho(1,0)}$ 
by Lemma~\ref{lem:common_support}. At this value of $\lambda'$, the 
right-hand side of Eq.~\eqref{eq:3_nondisparate} is 0 while the left-hand side is 
at least $(1-\epsilon)q(\lambda')>0$, yielding a contradiction.
\end{proof}

\begin{thm}\label{thm:qudits}
In an arbitrary quasiprobability representation, any set of 4 mutually 
non-orthogonal non-negative bases $\{\rho(\alpha,j):\alpha\in\mbb{Z}_4,j\in\mbb{Z}
_d\}$ must either:
\begin{itemize}
\item be disparate; or
\item satisfy the relationship
\begin{align}
\epsilon\rho(3,0) + \frac{1-\epsilon}{d}\unit = \sum_{\alpha=0}^2\sum_{j=1}^d 
p_{\alpha,j}\rho(\alpha,j)
\end{align}
for $\epsilon=\frac{1}{d+1}$ (up to a relabeling of basis states).
\end{itemize}
\end{thm}

\begin{proof}
Let $\{\rho(\alpha,j):\alpha\in\mbb{Z}_4,j\in\mbb{Z}_d\}$ be a set of four 
mutually non-orthogonal bases that are not disparate and are non-negative in 
some quasiprobability representation. 

As the bases are not disparate, then, relabeling the bases as necessary, we 
have $\rho(3,0)\in\mathcal{T}(\mc{A})$, where $\mc{A} = \{\rho(\alpha,j):
\alpha\in\mbb{Z}_3,j\in\mbb{Z}_d\}$, that is, the set of the elements of the 
first three bases, which must be disparate by 
Theorem~\ref{thm:qudits_coplanar}.

Therefore we can use the decomposition in Eq.~\eqref{eq:decompose_state} and 
relabel the bases such that $\vec{f} = (0,0,0)$ to obtain
\begin{align}
\epsilon\rho(3,0) + \frac{1-\epsilon}{d}\unit = \sum_{\alpha=0}^2\sum_{j=1}^d 
p_{0,j}\rho(\alpha,j)	\,.
\end{align}
As $\mu$ is convex-linear, we have
\begin{align}\label{eq:4_nondisparate}
\epsilon\mu_{\rho(3,0)}(\lambda) + \frac{1-\epsilon}{d}q(\lambda) = \sum_{
\alpha=0}^2\sum_{j=1}^d p_{\alpha,j}\mu_{\rho(\alpha,j)}(\lambda)
\end{align}
for all $\lambda\in\Lambda$. 

As $\rho(3,1)$ and $\rho(\alpha,j)$ are not orthogonal for any $\alpha\neq 3$ 
and any $j$, there exists $\lambda'\in\supps{\rho(3,1)}\cap\supps{\rho(\alpha,j)}$
by Lemma~\ref{lem:common_support}. At this value of $\lambda'$, the 
left-hand side of Eq.~\eqref{eq:4_nondisparate} is $\frac{1-\epsilon}{d}q(\lambda')>0$, 
while the right-hand side is at least $p_{\alpha,j}q(\lambda')>0$, so 
$p_{\alpha,j}\leq \frac{1-\epsilon}{d}$ for all $\alpha,j$.

As $\rho(3,0)$ and $\rho(\alpha,0)$ are not orthogonal for any $\alpha\neq 3$,
there exists $\lambda''\in\supps{\rho(3,0)}\cap\supps{\rho(\alpha,0)}$ by 
Lemma~\ref{lem:common_support}. At this value of $\lambda''$, the left-hand 
side of Eq.~\eqref{eq:4_nondisparate} is 
$(\frac{1-\epsilon}{d}+\epsilon)q(\lambda'')>0$, while the right-hand side 
is at most $\frac{1-\epsilon}{d}2q(\lambda'')>0$. 
Therefore $\epsilon\leq\frac{1-\epsilon}{d}$.

We now want to show that for each $\alpha$ there must exist a $k_{\alpha}$ 
such that $p_{\alpha,k_{\alpha}}=\frac{1-\epsilon}{d}$. To do this, let 
$\{\alpha,\beta,\gamma\}$ be a permutation of $\{0,1,2\}$. Then as $\rho(\beta,0)$ 
and $\rho(\gamma,0)$ are not orthogonal, there exists $\lambda_{\alpha}
\in\supps{\rho(\beta,0)}\cap\supps{\rho(\gamma,0)}$. At $\lambda_{\alpha}$, 
the right-hand side of Eq.~\eqref{eq:4_nondisparate} is 
$p_{\alpha,k_{\alpha}}q(\lambda_{\alpha})\leq\frac{1-\epsilon}{d}q(\lambda_{\alpha})$ 
for some value of $k_{\alpha}$. The left-hand side is at least 
$\frac{1-\epsilon}{d}q(\lambda_{\alpha})$, so we have that there exists 
$k_{\alpha}$ such that $p_{\alpha,k_{\alpha}} = \frac{1-\epsilon}{d}$. 
By considering all permutations of $\{0,1,2\}$, we see that this holds 
for all $\alpha$.

We can now show use this same approach to show that $\epsilon=\frac{1}{d+1}$. 
Let $\{\alpha,\beta,\gamma\}$ be a permutation of $\{0,1,2\}$. Then as $\rho(
\beta,k_{\beta})$ and $\rho(\gamma,k_{\gamma})$ are not orthogonal, there 
exists $\lambda'_{\alpha}\in\supps{\rho(\beta,k_{\beta})}\cap\supps{\rho(
\gamma,k_{\gamma})}$. At $\lambda'_{\alpha}$, the right-hand side of 
Eq.~\eqref{eq:4_nondisparate} is at least $\frac{1-\epsilon}{d}2q(\lambda'_{\alpha
})$. The only way the right-hand side can be equal to the left-hand side is 
if $\lambda'_{\alpha}\in\supps{\rho(3,0)}$ and $\epsilon = \frac{1-\epsilon}{d
}$, that is, $\epsilon = \frac{1}{d+1}$.
\end{proof}

We now show that in any quasiprobability distribution, any number $N$ of 
non-negative bases must still satisfy a symmetry constraint. However, as $N$ 
increases, this constraint becomes less restrictive.

\begin{thm}\label{thm:qudits_different}
In an arbitrary quasiprobability representation, any set of $N>4$ mutually 
non-orthogonal non-negative bases $\{\rho(\alpha,j):\alpha\in\mbb{Z}_N,j\in\mbb{Z}_d\}$ 
must either:
\begin{itemize}
\item be disparate; or
\item satisfy the relationship
\begin{align}
\epsilon\rho(3,0) + \frac{1-\epsilon}{d}\unit = \sum_{\alpha=0}^{N-2}\sum_{j=1
}^d p_{\alpha,j}\rho(\alpha,j)
\end{align}
for some $\epsilon\leq\frac{N-3}{N - 3 + d}$ (up to a relabeling of basis 
states).
\end{itemize}
\end{thm}

\begin{proof}
Let $\{\rho(\alpha,j):\alpha\in\mbb{Z}_N,j\in\mbb{Z}_d\}$ be a set of $N$ 
mutually non-orthogonal bases that are not disparate and are non-negative in 
some quasiprobability representation. 

Relabeling the bases as necessary, we have $\rho(M,0)\in\mathcal{T}(\mc{A})$, 
where $M=N-1$ and $\mc{A} = \{\rho(\alpha,j):\alpha\in\mbb{Z}_M,j\in\mbb{Z}_d
\}$, i.e., the set of the elements of the first $M$ bases.

Therefore we can use the decomposition in Eq.~\eqref{eq:decompose_state} and 
relabel the bases such that $\vec{f}$ is the zero vector to obtain
\begin{align}
\epsilon\rho(M,0) + \frac{1-\epsilon}{d}\unit = \sum_{\alpha\in\mbb{Z}_M}\sum_{j=1}^d p_{0,j}\rho(\alpha,j)	\,.
\end{align}
As $\mu$ is convex-linear, we have
\begin{align}\label{eq:N_nondisparate}
\epsilon\mu_{\rho(M,0)}(\lambda) + (1-\epsilon)\mu_{\frac{1}{d}\unit}(\lambda)
= \sum_{\alpha\in\mbb{Z}_M}\sum_{j=1}^d p_{\alpha,j}\mu_{\rho(\alpha,j)}(\lambda)
\end{align}
for all $\lambda\in\Lambda$. As $\rho(M,1)$ and $\rho(\alpha,j)$ are not 
orthogonal for all $\alpha\in\mbb{Z}_M$ and $j\in\mbb{Z}_d$, there exists 
$\lambda'\in\supps{\rho(M,1)}\cap\supps{\rho(\alpha,j)}$ by 
Lemma~\ref{lem:common_support}. At this value of $\lambda'$, the 
left-hand side of Eq.~\eqref{eq:N_nondisparate} is 
$\frac{1-\epsilon}{d}q(\lambda')>0$ while the right-hand side is at 
least $p_{\alpha,k}q(\lambda')$. Therefore 
$p_{\alpha,k}\leq\frac{1-\epsilon}{d}$ for all 
$\alpha\in\mbb{Z}_M$ and $k\in\mbb{Z}_d$.

As $\rho(M,0)$ and $\rho(\alpha,0)$ are not orthogonal for all $\alpha\in\mbb{
Z}_M$, there exists $\lambda''\in\supps{\rho(M,0)}\cap\supps{\rho(\alpha,0)}$ 
by Lemma~\ref{lem:common_support}. At this value of $\lambda''$, the 
left-hand side of Eq.~\eqref{eq:N_nondisparate} is $(\frac{1-\epsilon}{d}+\epsilon)
q(\lambda'')>0$ while the right-hand side is at most $(M-1)\frac{1-\epsilon}{d
} q(\lambda'')$.
\end{proof}

\subsection{Upper bound on the number of non-negative bases for qudits}

Theorems~\ref{thm:qudits_coplanar}--\ref{thm:qudits_different} provide strong 
constraints on the relation between any set of non-negative bases in a 
quasiprobability representation. However, it is unclear how to use these 
theorems to obtain an upper bound on the number of non-negative bases in a 
quasiprobability representation. In order to obtain an upper bound (which 
will not be tight), we change tack and exploit the fact that for all 
$\lambda\in\Lambda$ there exists an operator $F(\lambda)$ acting on $\mc{H}_d$ 
such that
\begin{align}\label{eq:convex_operators}
\mu_{\rho}(\lambda) = \tr{\rho F(\lambda)}
\end{align}
for all $\rho\in\mc{B}(\mc{H}_d)$. This will enable us to show that there are 
no more than $2^{d^2}$ states that are elements of a non-negative basis in 
any quasiprobability distribution, without requiring that the bases are 
mutually non-orthogonal. To obtain this bound, we note that any density 
matrix can be written as a linear combination of the $F(\lambda)$~\cite{Montina2006},
so the $F(\lambda)$ must be a basis for the space of operators 
acting on $\mc{H}_d$.

\begin{thm}\label{thm:qudit_final}
For any quasiprobability representation of $\mc{H}_d$, there are no more than 
$2^{d^2}$ states that are elements of non-negative bases.
\end{thm}

\begin{proof}
Let $\{\Pi_{\alpha}:\alpha\in\mbb{Z}_{d^2}\}$ be a trace-orthonormal basis of 
$\mc{B}(\mc{H}_d)$ and let $\{F(\lambda_{\beta}):\beta\in\mbb{Z}_{d^2}\}$ be 
a set of $d^2$ linearly independent operators, which must exist as the set 
$\{F(\lambda)\}$ is a basis for the space of operators acting on $\mc{H}_d$. 
Then, for all $\rho\in\mc{B}(\mc{H}_d)$ and $\beta\in\mbb{Z}_{d^2}$, we can 
write
\begin{align}
\rho &= \sum_{\alpha\in\mbb{Z}_{d^2}} g_{\alpha}(\rho)\Pi_{\alpha}	\,,
\nonumber\\
F(\lambda_{\beta}) &= \sum_{\alpha\in\mbb{Z}_{d^2}} f_{\alpha,\beta}\Pi_{
\alpha}	\,,
\end{align}
where 2 states $\rho,\rho'\in\mc{B}(\mc{H}_d)$ have the same coefficients 
$g_{\alpha}$ for all $\alpha$ if and only if $\rho=\rho'$. Therefore we can 
rewrite Eq.~\eqref{eq:convex_operators} as
\begin{align}
\mu_{\rho}(\lambda_{\beta}) = \sum_{\alpha} f_{\alpha,\beta}g_{\alpha}(\rho)	
\,,
\end{align}
As the $\{F(\lambda_{\beta}):\beta\in\mbb{Z}_{d^2}\}$ are linearly 
independent, $f_{\alpha,\beta}$ must be invertible. Therefore for any set of 
values $\{\mu(\lambda_{\beta}):\beta\in\mbb{Z}_{d^2}\}$, there can be at most 
one state $\rho$ such that $\mu_{\rho}(\lambda_{\beta})=\mu(\lambda_{\beta})$ 
for all $\beta\in\mbb{Z}_{d^2}$. 

From Lemma~\ref{lem:sum_prob}, any state that is an element of a non-negative 
basis can only assign one of two values to any point $\lambda\in\Lambda$, 
namely, 0 or $q(\lambda)$. Therefore there are only $2^{d^2}$ possible sets 
of values of $\mu$ over $\{\lambda_{\beta}:\beta\in\mbb{Z}_{d^2}\}$ that 
correspond to elements of non-negative bases.
\end{proof}

The bound on the number of states that are elements of a non-negative basis 
in Theorem~\ref{thm:qudit_final} is not tight. For example, not all vectors $g$ 
correspond to a valid density operator. In particular, no quantum state can 
have $g_{\alpha}=0$ for all $\alpha$. Furthermore, if an element of a 
non-negative basis assigned nonzero probability to more than $d^2-d+1$ of the 
points $\{\lambda_{\beta}\}$, then as the elements of a non-negative basis 
have disjoint support, at least one of the other elements of a non-negative 
basis would have to assign zero probability to all of the points $\{\lambda_{\beta}\}$
and so would have $g_{\alpha}=0$ for all $\alpha$. Therefore all 
non-negative states assign nonzero probability to between 1 and $d^2 - d + 1$ 
of the points $\{\lambda_{\beta}\}$.

However, even accounting for this does not substantially decrease the upper 
bound. Furthermore, for qubits, we proved in Theorem~\ref{thm:qubit_new} that 
no more than 8 states can be elements of a non-negative basis. The upper 
bound from Theorem~\ref{thm:qudit_final} is 16 states, and even excluding the 
combinations of $\mu$ discussed above only reduces the upper bound to 14 
states.

\section{Discussion and conclusion}\label{negativity:conclusion}

We have shown that for any quasiprobability representation of a qubit, any 
three non-negative bases cannot be coplanar in the Bloch sphere (i.e., they 
must be disparate). Moreover, if there are four non-negative bases, then they 
must correspond to the vertices of a right cuboid circumscribed by the Bloch 
sphere. We provided an exhaustive list of all ``classical'' subtheories of a 
qubit that include states, measurements and nontrivial transformations. 
These cases revealed several interesting features. Both families of three 
bases that are permuted by a nontrivial unitary group can only be 
non-negative in a quasiprobability representation when they are sufficiently 
``far'' from being coplanar. However, there is an exception to this behavior, 
as the bases in Eq.~\eqref{eq:D3_bases} are non-negative in some 
quasiprobability representation even when $\theta\to0$ (i.e., the three bases 
are almost degenerate). We have also found that whenever a subtheory of qubit 
states and measurements are non-negative in some quasiprobability 
representation, there exists a quasiprobability representation in which all 
unitary transformations that permute non-negative states correspond to a 
permutation of the ontic states. 

While we have primarily focused on the qubit case, we have also shown that 
the results for qubits directly generalize in that any three mutually 
non-orthogonal non-negative bases in a quasiprobability distribution must be 
disparate and any four or more mutually non-orthogonal non-negative bases 
must either by disparate or satisfy a symmetry constraint. In this sense, 
quantum states and measurements with a small amount of complementarity can be 
quite difficult to model in a classical theory. In addition, we have obtained 
an upper bound of $2^{d^2}$ on the number of states that are elements of a 
non-negative basis. 

We conclude with some discussion of the implications of our results for 
quantum computation, and some future research directions. While our results 
have been presented in the context of single qudits, they are equally 
applicable to multiple qudit systems. Our upper bound on the number of 
non-negative basis states of a qudit, although quite loose, suggests that 
universal quantum computation leads to negativity in \textit{any} 
quasiprobability distribution. This matches the intuition obtained from the 
specific case of the single qudit discrete Wigner function~\cite{Cormick2006}.

While our higher-dimensional results can be applied to quasiprobability 
representations of multiple qubit systems, it is not clear how this approach 
accords with classical simulations of quantum systems. In particular, 
multi-qubit stabilizers can be efficiently simulated classically~\cite{Aaronson2004}
and yet do not correspond to a set of non-negative bases in any 
quasiprobability representation. To see this, note that stabilizer states and 
$X$ and $Y$ measurements (i.e., in bases corresponding to stabilizer states) 
can lead to violations of a Bell inequality~\cite{Mermin1990} and so cannot 
admit a locally causal model.

A natural way of generalizing a quasiprobability representation for a single 
qubit to one for multiple qubits is to take tensor products of the operators
$\{F(\lambda)\}$ and $\{G(\lambda)\}$ that define the single qubit 
quasiprobability representation via Eq.~\eqref{eq:quasi_F} and 
\eqref{eq:quasi_G}. By construction, such a quasiprobability representation 
will be non-negative for all tensor products of the single qubit states
with non-negative distributions, but may also be non-negative for 
other bases that include entangled states. For the entangled states to 
be accessible in a classical subtheory, there must be some unitary that 
permutes non-negative bases and maps a non-negative product basis to a
non-negative basis that contains an entangled state. Such a unitary can 
only be viewed as supervening on a permutation of ontic states (which 
could always be done for a single qubit) if it leaves the set of tensor 
products of the $\{F(\lambda)\}$ invariant under conjugation. Unfortunately,
it is unclear whether such unitaries exist for any set of operators that
define our single qubit quasiprobability representations, although based
on the results of \cite{Clark2007} we have some evidence to suggest that
they do not.

\begin{acknowledgments}
We acknowledge helpful discussions with Chris Ferrie, Chris Fuchs, Markus 
Mueller and Rob Spekkens, and financial support from the Australian Research 
Council via the Centre of Excellence in Engineered Quantum Systems (EQuS), 
project number CE11001013. 
\end{acknowledgments}


\begin{thebibliography}{99}
\bibitem{Wigner1971} E.~Wigner, in Perspectives in Quantum Theory, edited by 
W. Yourgrau and A. Van der Merwe (MIT Press, Cambridge, 1971), pp. 2536.
\bibitem{Ferrie2011} C. Ferrie, Rep. Prog. Phys. \textbf{74}, 116001 (2011).
\bibitem{Kenfack2004} A.~Kenfack and K.~\.Zyczkowski, J. Opt. B: Quantum 
Semiclass. Opt. \textbf{6}, 396 (2004). 
\bibitem{Hudson1974} R.~L.~Hudson, Rep. Math. Phys. \textbf{6}, 249 (1974).
\bibitem{Soto1983} F.~Soto and P.~Claverie, J. Math. Phys. \textbf{24}, 97 (1983).
\bibitem{Gross2006} D.~Gross, J. Math. Phys. \textbf{47}, 122107 (2006).
\bibitem{Gross2007} D.~Gross, Appl. Phys. B \textbf{86}, 367 (2007).
\bibitem{Spekkens2008} R.~W.~Spekkens, Phys. Rev. Lett. \textbf{101}, 020401 (2008).
\bibitem{Veitch2012} V. Veitch, C. Ferrie, and J. Emerson, eprint arXiv:1201.1256 (2012).
\bibitem{Spekkens2005} R.~W.~Spekkens, Phys. Rev. A \textbf{71}, 052108 (2005).
\bibitem{Shor2008} P.~W.~Shor and S.~P.~Jordan, Quant. Inf. Comp. \textbf{8}, 
681 (2008).
\bibitem{Ferrie2008} C.~Ferrie and J.~Emerson, J. Phys. A: Math. Gen. \textbf{41}, 352001 (2008).
\bibitem{bb_model} E.~G.~Beltrametti and S.~Bugajski, J. Phys. A: Math. Gen. 
\textbf{28}, 3329 (1995).
\bibitem{Montina2006} A.~Montina, Phys. Rev. Lett. \textbf{97}, 180401 (2006).
\bibitem{Bartlett2012} S.~D.~Bartlett, T.~Rudolph and R.~W.~Spekkens, Phys. 
Rev. A \textbf{86}, 012103 (2012).
\bibitem{Gibbons2004} K.~S.~Gibbons, M.~J.~Hoffman and W.~K.~Wootters, Phys. 
Rev. A \textbf{70}, 062101 (2004).
\bibitem{Spekkens2007} R.~W.~Spekkens, Phys. Rev. A \textbf{75}, 032110 (2007).
\bibitem{Weedbrook2012} C.~Weedbrook, S.~Pirandola, R.~Garcia-Patron, 
N.~J.~Cerf, T.~C.~Ralph, J.~H.~Shapiro and S.~Lloyd, \rmp \textbf{84}, 621 (2012).
\bibitem{Bennett1993} C.~H.~Bennett, G.~Brassard, C.~Cr\'epeau, R.~Jozsa, 
A.~Peres and W.~K.~Wootters, Phys. Rev. Lett. {\bf 70}, 1895 (1993)
\bibitem{Bennett1992} C.~H.~Bennett, S.~J.~Wiesner, Phys. Rev. Lett. {\bf 69}, 2881 (1992).
\bibitem{Ferrie2009} C.~Ferrie and J.~Emerson, New J. Phys. \textbf{11}, 
063040 (2009).
\bibitem{Tinkham2003} M.~Tinkham2003, \textit{Group Theory and Quantum 
Mechanics} (Dover Publications, New York, U.S.A, 2003).
\bibitem{Cormick2006} C.~Cormick, E.~F.~Galv\~ao, D.~Gottesman, J.~P.~Paz and 
A.~O.~Pittenger, Phys. Rev. A \textbf{73}, 012301 (2006).
\bibitem{Aaronson2004} S.~Aaronson and D.~Gottesman, Phys. Rev. A \textbf{70},
052328 (2004).
\bibitem{Mermin1990} N.~D.~Mermin, Phys. Rev. Lett. \textbf{65}, 1838 (1990).
\bibitem{Clark2007} S.~Clark, R.~Jozsa and N.~Linden, eprint arXiv:quant-ph/0701103 (2007).
\end{thebibliography}
\end{document}